\documentclass[11pt,reqno,letter]{amsart}

\usepackage{amssymb}
\usepackage[dvipdfmx]{graphicx}

\usepackage[usenames]{color}
\usepackage{caption}
\usepackage[dvips]{psfrag}
\usepackage{latexsym}
\usepackage{amsmath}
\usepackage{amsthm}
\usepackage{tabularx}
\usepackage{bm}
\usepackage{lscape}
\usepackage[numbers]{natbib}
\usepackage[usenames]{color}
\usepackage{url}
\usepackage{hyperref}
\usepackage{upgreek}
\usepackage{units}

\usepackage[normalem]{ulem}

\theoremstyle{plain}
\newtheorem{theorem}{Theorem}[section]
\newtheorem{lemma}[theorem]{Lemma}

\newtheorem{corollary}[theorem]{Corollary}
\newtheorem{proposition}[theorem]{Proposition}

\theoremstyle{definition}
\newtheorem{definition}[theorem]{Definition}
\newtheorem{remark}[theorem]{Remark}
\newtheorem{example}[theorem]{Example}
\newtheorem{assumption}[theorem]{Assumption}

\numberwithin{equation}{section}
\numberwithin{theorem}{section}
\numberwithin{table}{section}
\numberwithin{figure}{section}

\DeclareMathOperator*{\argmax}{arg\,max}
\DeclareMathOperator*{\argmin}{arg\,min}

\title{Worst-case sensitivity}
\date{\today}
\author[Gotoh]{Jun-ya Gotoh\textsuperscript{$\dagger$}}

\author[Kim]{Michael Jong Kim\textsuperscript{$\ddagger$}}

\author[Lim]{Andrew E.B. Lim\textsuperscript{$*$}}

\dedicatory{\textsuperscript{$\dagger$}Department of Industrial and Systems Engineering, Chuo University, Tokyo, Japan. Email: jgoto@indsys.chuo-u.ac.jp \\ \textsuperscript{$\ddagger$}Sauder School of Business, University of British Columbia, Vancouver, Canada. Email: mike.kim@sauder.ubc.ca \\ \textsuperscript{$*$}Department of Analytics and Operations, Department of Finance, and Institute for Operations Research and Analytics, National University of Singapore, Singapore. Email: andrewlim@nus.edu.sg
}

\begin{document}

\begin{abstract}
We introduce the notion of {\it Worst-Case Sensitivity}, defined as the worst-case rate of increase in the expected cost of a Distributionally Robust Optimization (DRO) model when the size of the uncertainty set vanishes. We show that worst-case sensitivity is a {\it generalized measure of deviation} and that a large class of DRO models are essentially mean-(worst-case) sensitivity problems when uncertainty sets are small, unifying recent results on the relationship between DRO and regularized empirical optimization with worst-case sensitivity playing the role of the regularizer. More generally, DRO solutions can be sensitive to the family and size of the uncertainty set, and reflect the properties of its worst-case sensitivity. We derive closed-form expressions of worst-case sensitivity for well known uncertainty sets including smooth $\phi$-divergence, total variation, ``budgeted" uncertainty sets, uncertainty sets corresponding to a convex combination of expected value and CVaR, and the Wasserstein metric. These can be used to select the uncertainty set and its size for a given application.
\end{abstract}

\maketitle


\section{Introduction}

Consider a cost function $f(x, Y)$ where $x$ is the decision and $Y$ is a discrete random variable with (nominal) distribution ${\mathbb P}=[p_1,\cdots, p_n]$. For example, $\mathbb P$ could be the empirical distribution associated with a historical sample of $Y$'s generated {\it iid} from some unknown distribution. We would like to find a decision that performs well out-of-sample, a candidate for which is the minimizer of the sample average approximation (SAA)
\begin{align*}
\min_x ~ {\mathbb E}_{\mathbb P}[f(x, Y)].
\end{align*}
In many situations, the solution of this problem does not perform well out-of-sample because the in-sample model is misspecified. If the expected cost of the in-sample optimizer is sensitive to perturbations of this model, the out-of-sample expected cost may increase significantly because of differences between the in- and out-of-sample distributions. One approach to account for model misspecification is Distributionally Robust Optimization (DRO), where decisions are obtained by optimizing the worst-case expected cost over a family of alternative models. More generally, we would like to find a decision that performs well (in-sample) and continues to perform well when the in-sample model is incorrect.

\subsubsection*{Distributionally Robust Optimization (DRO)} Let ${\mathcal Q}(\varepsilon)$ be an {\it uncertainty set} of size $\varepsilon$, specifically,  a set of probability distributions containing the nominal distribution $\mathbb P$ that is increasing in $\varepsilon$ and degenerates to the nominal distribution ${\mathcal Q}(0) = \{{\mathbb P}\}$ when $\varepsilon=0$. For example, ${\mathcal Q}(\varepsilon)$ could be a set of the form $\{{\mathbb Q}\,| d({\mathbb Q}\,|{\mathbb P})\leq \varepsilon\}$ where $d({\mathbb Q}\,|{\mathbb P})$ is  $\phi$-divergence or the Wasserstein metric. The worst-case expected cost with respect to ${\mathcal Q}(\varepsilon)$ is
\begin{align}
V(\varepsilon; f(x,\cdot)) &:=
\max_{{\mathbb Q}\in{\mathcal Q}(\varepsilon)} \mathbb{E}_{\mathbb{Q}}[f(x, Y)].
\label{eq:V}
\end{align}
Distributionally Robust Optimization (DRO) is the worst-case problem
\begin{align*}
\min_x V(\varepsilon; f(x,\cdot)) &\equiv \min_x \max_{{\mathbb Q}\in{\mathcal Q}(\varepsilon)} \mathbb{E}_{\mathbb{Q}}[f(x, Y)].
\end{align*}
We refer to $V$ as the {\it value function} of the worst-case problem which, under the assumptions in this paper, is increasing, continuous and concave in $\varepsilon$. When it is clear from the context, we write $V(\varepsilon) \equiv V(\varepsilon; f(x,\cdot))$ if $f$ and $x$ are fixed and we are concerned about the dependence of $V$ on $\varepsilon$, or $V(\varepsilon, x)\equiv V(\varepsilon; f(x,\cdot))$ if we are interested in an optimization problem over $x$. The solution of the DRO problem is denoted by $x(\varepsilon)$.

\subsubsection*{Worst-case Sensitivity} Suppose $x$ is fixed and $V(\varepsilon, x)$ has a finite right derivative in $\varepsilon$  at $\varepsilon=0$. The {\it worst-case sensitivity} is the right derivative of the value function at $\varepsilon=0$:
\begin{align}
{\mathcal S}_{\mathbb P}[f(x, \cdot)] &= V'(0^+; f(x,\cdot)) = \lim_{\varepsilon\downarrow 0}\frac{V(\varepsilon; f(x, \cdot))-\sum_{i=1}^n p_i f(x, Y_i)}{\varepsilon}.
\label{eq:sensitivity-general1}
\end{align}
When the sensitivity is large, small deviations from the nominal distribution can result in a large increase in the expected cost; such a decision is not robust.

\subsection*{Summary of contributions}
We show under mild conditions that a large class of DRO problems can be interpreted as multi-objective problems that tradeoff between expected cost and some measure of sensitivity. This measure of sensitivity is bounded above by worst-case sensitivity \eqref{eq:sensitivity-general1} where the bound is tight when $\varepsilon$ vanishes. This shows that DRO is fundamentally a tradeoff between mean and sensitivity, generalizing an interpretation of the ``variance regularizer" from \cite{gotoh2018robust} where the penalty form of DRO with smooth $\phi$-divergence was shown to be equivalent to a mean-variance problem when uncertainty sets were small, and unifying recent results connecting various DRO models and ``regularized SAA".

We derive explicit expressions for worst-case sensitivity \eqref{eq:sensitivity-general1} for a number of  popular uncertainty sets including smooth $\phi$-divergence, total variation, ``budgeted uncertainty" (i.e. hard constraints on the likelihood ratio), and the Wasserstein distance, which are summarized in Table \ref{table:summary}.  Under standard assumptions, worst-case sensitivity is a {\it generalized measure of deviation} \cite{rockafellar2006generalized}. Intuitively, sensitivity is large if small errors in the nominal model---particularly, the probability of extreme costs---have a big impact on the mean, which will be the case if the spread of the cost distribution is large. Sensitivity can be reduced by selecting a decision with a smaller spread, and the equivalence between DRO mean-sensitivity (spread) optimization shows that this is precisely what it is doing. Different uncertainty sets correspond to different measures of spread, which determines the nature of the DRO solutions, while the tradeoff between expected cost and sensitivity is determined by the size of the uncertainty set.

\begin{table}
\begin{center}
\caption{Worst case sensitivities}
\label{table:summary}
\begin{tabular}{|l|l|l|l|}
\hline
Type of divergence/worst-case objective & Worst-case sensitivity \\
\hline
$\phi$-divergence with any smooth $\phi$-function & $\sqrt{\frac{2{\mathbb V}_{\mathsf{p}}(\mathsf{f})}{\phi''(1)}}$ \\ [2pt]
Total variation (or $\phi$-div. with $\phi(z)=|z-1|$) & $\frac{1}{2}\big(\max(\mathsf{f})-\min(\mathsf{f})\big)$\\ [2pt]
budgeted uncertainty (or $\phi$-div. with $\phi=\delta_{[0,1+\varepsilon]}$) & $\mathbb{E}_{\mathsf{p}}(\mathsf{f})-\min(\mathsf{f})$ \\[2pt]
$(1-\delta)$``expected cost''$+\delta$``max cost''
 & $\max(\mathsf{f})-\mathbb{E}_{\mathsf{p}}(\mathsf{f})$\\ [5pt]
$\phi$-divergence with $\phi=\delta_{[\frac{1}{1+\varepsilon},1+\varepsilon]}$
 & $\mathrm{CVaR}_{\mathsf{p},\frac{1}{2}}(\mathsf{f})-\mathbb{E}_{\mathsf{p}}(\mathsf{f})$ \\[2pt]
$(1-\delta)$``expected cost''$+\delta$``$\alpha$-CVaR'' & $\mathrm{CVaR}_{\mathsf{p},
\alpha}(\mathsf{f})-\mathbb{E}_{\mathsf{p}}
(\mathsf{f})$\\ [5pt]
Wasserstein distance with $f(z_i)-f(Y_i)\leq L\|z_i-Y_i\|_p$ & $\max\limits_{i=1,\cdots,n} \max\limits_{z_i}\frac{f(z_i)-f(Y_i)}{\|z_i-Y_i\|_{\mathit{p}}}$ \\ [2pt]
\hline
\end{tabular}
\end{center}
\end{table}

The closed-form expressions we derive for worst-case sensitivity can be used to {\it select the uncertainty set} that is most appropriate for any given application. For example, DRO with  ``budgeted uncertainty"  controls the spread of the ``good" side of the cost distribution (Table \ref{table:summary}). Doing so, however,  may result in solutions that increase the length of the ``bad" side of the tail, making it a poor choice in applications where losses can be large. Indeed, it is possible for the DRO solution under one uncertainty set (e.g., ``budgeted uncertainty") to be less robust than SAA from the perspective of another (e.g., smooth $\phi$-divergence); an inappropriate choice of uncertainty set can result in the SAA optimizer being replaced by a decision that is even less robust. Finally, for the inventory problem, $L_1$  Wasserstein sensitivity is independent of the decision, which makes it a poor choice of uncertainty set for this application.

\subsection*{Overview of paper}
We review the relevant literature in Section \ref{sec:LitRev}. In Section \ref{sec:DRO_sensitivity}, we show how well-known duality results for DRO allow us to view all DRO problems as  a tradeoff between mean cost and worst-case sensitivity when uncertainty sets are small. Under standard assumptions, these sensitivity measures are also measures of spread. We derive explicit expressions for worst-case sensitivity for a number of different uncertainty sets in Section \ref{sec:WCS}. Examples are presented in Section \ref{sec:Examples}.

\remark
There are uncertainty sets where the right derivative $V'(0^+)$ of the worst-case problem is unbounded. For example, when $d({\mathbb Q}\,|\,{\mathbb P})$ is sufficiently smooth $\phi$-divergence, $V(\varepsilon)-V(0)$ is $O(\sqrt \varepsilon)$ and $V'(0^+)$ is unbounded, so the definition \eqref{eq:sensitivity-general1} is not helpful. If $V(\varepsilon) - V(0) \sim O(g(\varepsilon))$, where $g(\varepsilon)$ is strictly concave and increasing in $\varepsilon$ with $g(0)=0$, we define worst-case sensitivity (with growth rate $g(\varepsilon)$) as
\begin{align}
{\mathcal S}_{\mathbb P}[f] &=  \lim_{\epsilon\downarrow 0}\frac{V(\varepsilon)-\sum_{i=1}^n p_i f(Y_i)}{g(\varepsilon)},
\label{eq:sensitivity-general2}
\end{align}
in which case
\begin{align}
V(\varepsilon, x) &= \sum_{i=1}^n p_i f(x,\,Y_i) + g(\varepsilon)\,{\mathcal S}_{\mathbb P}[f(x,\cdot)] + o(g(\varepsilon)), \; \varepsilon>0.
\label{eq:vf-theta2}
\end{align}
As in \eqref{eq:vf-theta1}, optimizing the worst-case cost $V(\varepsilon, x)$ is again a tradeoff between expected cost and sensitivity, the only difference being that that robustness parameter $\varepsilon$ no longer appears linearly. Since $V(\varepsilon)$ is concave and increasing, so too is $g(\varepsilon)$, and  $g(0)=0$.

\section{Literature Review}
\label{sec:LitRev}

While the DRO (and robust optimization) literature is very large, the focus has been on methodology for solving worst-case problems for a diverse range of uncertainty sets and applications \cite{ben1989,ben2013robust,ber2004,del2010,hansen2008robustness,lim2007pricing,george2006ROsurvey,pet2000,Guzin2019newsvendorTV}. While the motivation for DRO is to find decisions that are insensitive to model uncertainty, worst-case solutions can be sensitive to the choice of the uncertainty set, and ``robust" solutions under one uncertainty set can be less robust than the SAA solution under another. There is little  guidance on how the  uncertainty set should be chosen for a given application.

When it comes to size, much of the literature suggests that $\varepsilon$ be chosen so that the uncertainty set ${\mathcal Q}(\varepsilon)$ contains the true model with high probability (e.g., $95\%$). This ignores the multi-objective nature of DRO, and  there is also little reason to believe that a pre-ordained confidence level selected independently of data and objective always leads to a ``good" decision for all applications\footnote{Indeed, extensive calculations with a state of the art data science methods and otherworldly computational power  \cite{adams1995hitch} suggests that the answer is more likely to be $42\%$.}. A second approach treats $\varepsilon$ as a free parameter, like the regularization parameter in regression, chosen to optimize an estimate of out-of-sample expected cost obtained from cross-validation or the bootstrap. While this accounts for both data and the objective, it ignores the sensitivity reduction objective intrinsic to DRO. It is easy to construct examples with high levels of model uncertainty where $\varepsilon=0$ (SAA) optimizes the resampled estimate of the out-of-sample cost \cite{lim2020calibration}. If sensitivity is ignored, this approach recommends that it is optimal not to be robust.

Worst-case sensitivity \eqref{eq:sensitivity-general1} is defined in \cite{lam2016robust} where it is used to study the accuracy of simulated estimates of the mean of a random variables. However, this paper does not consider implications for robust optimization, and only considers uncertainty sets defined in terms of relative entropy (a special case of smooth $\phi$-divergence).

The connection between worst-case sensitivity and ``minimally robust"  DRO with smooth $\phi$-divergence penalty functions is discussed in  \cite{gotoh2018robust,lim2020calibration}, with \cite{gotoh2018robust} showing that worst-case sensitivity corresponds to the variance of the reward, and \cite{lim2020calibration} studying the out-of-sample properties of the associated mean-variance frontier. The paper \cite{Duchi2017variance} derives high-probability performance bounds for the out-of-sample expected reward  for solutions of variance-regularized loss minimization.
The present paper derives worst-case sensitivity {\it for uncertainty sets beyond smooth $\phi$-divergence}, provides elementary arguments linking DRO and mean-sensitivity optimization under very mild assumptions, and shows that worst-case sensitivity is generally a measure of spread. The mean-sensitivity connection shows that DRO is intrinsically a tradeoff between minimizing the expected cost and controlling the spread of the cost distribution under the nominal, with different uncertainty sets giving rise to different measures of spread. Closed form expressions for worst-case sensitivity make this tradeoff explicit, and can also be used to select the family and size of the uncertainty set in any given application.

Several papers discuss the relationship between DRO and ``regularized SAA". For example, \cite{blanchet2019,gao2017wasserstein,kuhn2020WassersteinSurvey,abadeh2015distributionally} show that worst-case regression and classification problems with the Wasserstein metric are equivalent to regularized versions of these problems.
More generally,   \cite{bertsimas2018,blanchet2019,Duchi2017variance,gao2017wasserstein,kuhn2020WassersteinSurvey,abadeh2015distributionally,Xu2010} consider worst-case optimization for a particular uncertainty set and derive the corresponding ``regularized SAA" problem. We unify these ideas by showing that the ``regularizer" is worst-case sensitivity. We also show that it is a ``generalized measure of deviation" \cite{rockafellar2006generalized} and derive an explicit expression for several popular uncertainty sets. These expressions show how an uncertainty set affects the solution of a DRO model, and can be used to select the family and size of the uncertainty set in a given application.

Variance (sensitivity) reduction properties of DRO solutions have also been observed empirically in the literature. In a network model of project management \cite{karthik2019crashing}, simulations show that robustness prioritizes reducing the variance of the cost (activity duration) over its mean, while \cite{kim2015MAB} shows variance reduction of the out-of-sample cost under the solution of a dynamic DRO problem. The solution of a worst-case newsvendor application \cite{karthik2020inventory} is shown to be optimal for a risk-neutral model with a heavy-tailed demand distribution for the demand, thus controlling the impact of extreme events on the cost.

As a final note, while DRO controls the spread (sensitivity) of the cost, robustness can reduce but also  increase the out-of-sample variance of the solution
(e.g. solution variability decreases for robust classification and regression models that are equivalent to $L_p$ solution regularization \cite{bertsimas2018,blanchet2019,gao2017wasserstein,abadeh2015distributionally,Xu2010} while \cite{gotoh2018robust} gives an example where solution variability increases).

\section{DRO and sensitivity}
\label{sec:DRO_sensitivity}

We reinterpret  duality results from the perspective of mean-sensitivity tradeoffs and show, under mild conditions, that DRO problems are mean-sensitivity problems, that worst-case sensitivity is a tight upper bound of ``DRO sensitivity" being equal in the limit $\varepsilon\downarrow 0$, and that worst-case sensitivity is a measure of the spread of the cost distribution. The purpose of this section is to motivate our study of worst-case sensitivity by putting it in the context of classical results and highlighting  the cost-sensitivity tradeoff is intrinsic to DRO.

\subsection*{Mean-sensitivity problems}
Recall the worst-case objective \eqref{eq:V}. Since uncertainty sets ${\mathcal Q}(\varepsilon)$ are increasing in $\varepsilon$ and contain the nominal $\mathbb P$, the worst-case expected cost is equal to the nominal expected cost when $\varepsilon=0$ and monotonically increasing in $\varepsilon$ (see Figure \ref{fig:sensitivity-plots1}). It follows that there is a function ${\mathcal A}\big(\varepsilon; f(x, Y)\big)$, which we refer to as the {\it ambiguity cost}, that is non-negative and increasing in $\varepsilon$ such that ${\mathcal A}\big(0; f(x, Y)\big)=0$ and
\begin{align}
V(\varepsilon, x) &= {\mathbb E}_{\mathbb P}[f(x, Y)] + {\mathcal A}\big(\varepsilon; f(x, Y)\big)
\label{eq:DRO mean ambiguity-cost}
\end{align}

\begin{figure}[h]
\centering
\includegraphics[scale=0.45]{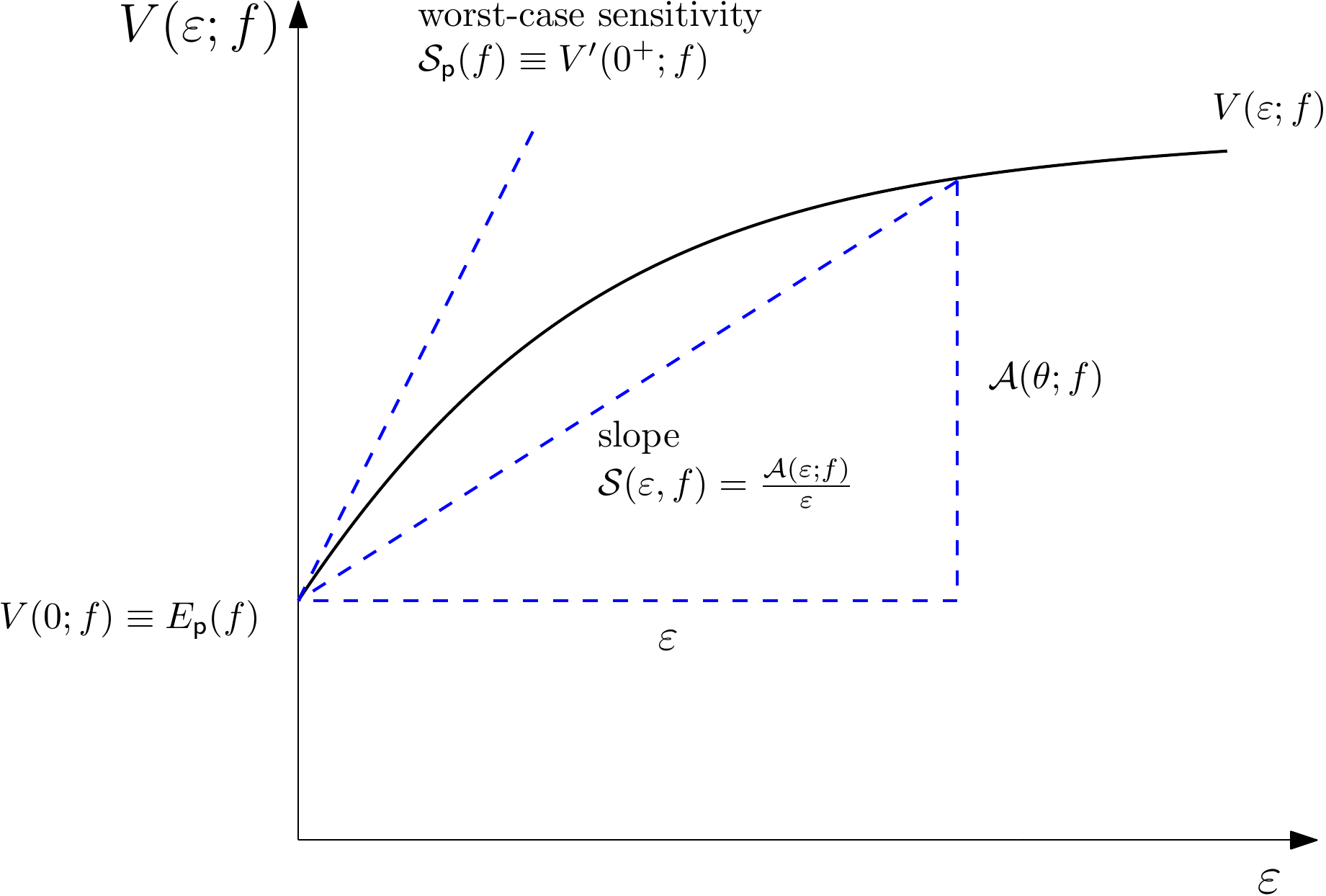}
\caption{Worst-case sensitivity is the right derivative of $V(\varepsilon; f)$ at $\varepsilon=0$.}
\label{fig:sensitivity-plots1}
\end{figure}

For every $\varepsilon$, we define {\it average sensitivity}
\begin{align*}
{\mathcal S}\big(\varepsilon; f(x, Y)\big) & \equiv \frac{1}{\varepsilon}{\mathcal A}\big(\varepsilon; f(x, Y)\big).
\label{eq:DRO mean ambiguity-cost2}
\end{align*}
It follows that the worst-case problem is a mean-sensitivity problem:
\begin{align}
\min_x V(\varepsilon,\,x) & \equiv \min_x {\mathbb E}_{\mathbb P}[f(x, Y)] + \varepsilon {\mathcal S}\big(\varepsilon; f(x, Y)\big).
\end{align}
When $\varepsilon=0$, the robust decision maker optimizes the SAA. As $\varepsilon$ increases, he/she absorbs a larger expected cost in return for a lower sensitivity, as illustrated in Figure \ref{fig:sensitivity-plots3}.

\begin{figure}[h]
\centering
\includegraphics[scale=0.25]{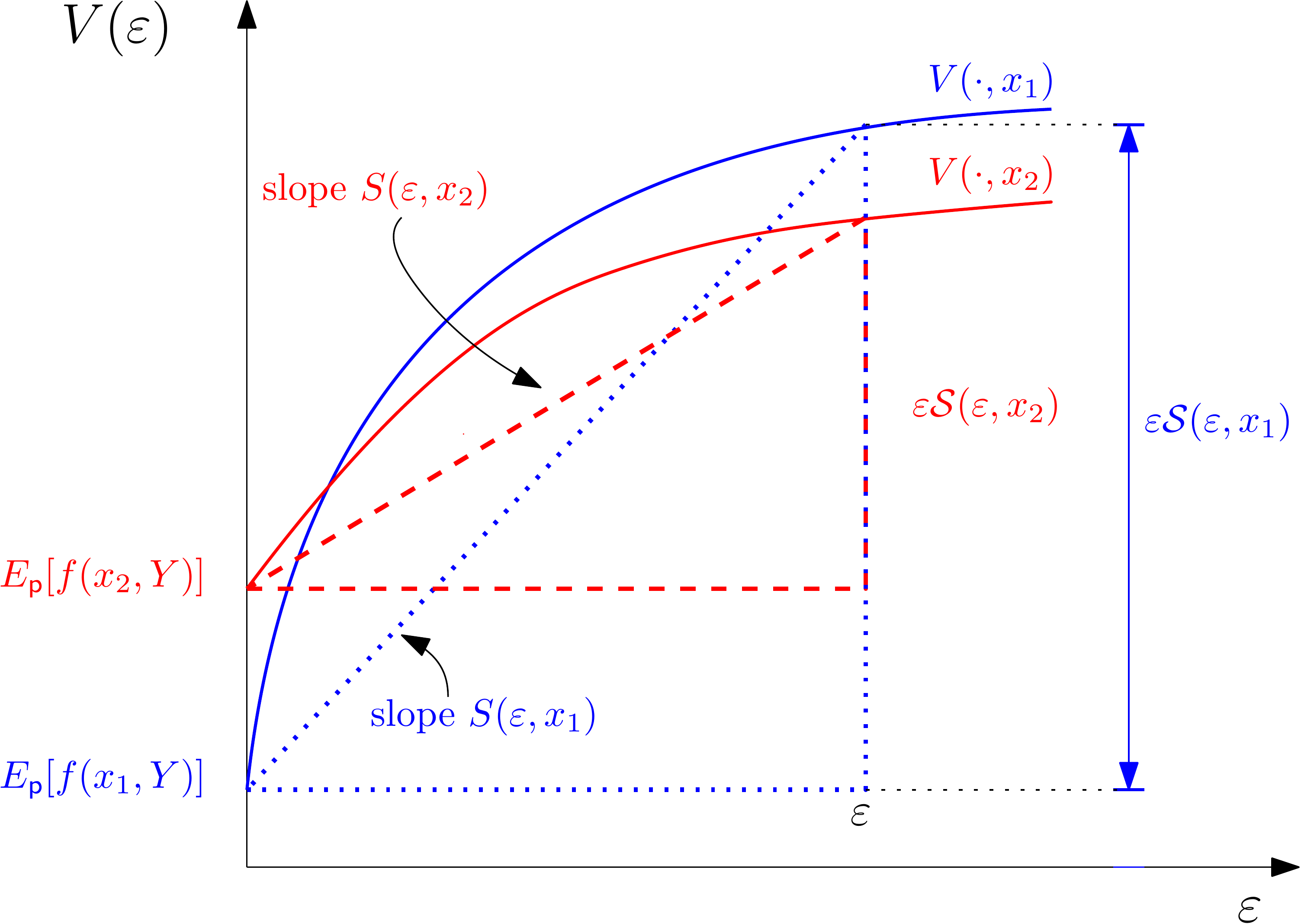}
\caption{DRO is a tradeoff between expected cost ${\mathbb E}_{\mathbb P}[f( x, Y)]$ and average sensitivity ${\mathcal S}(\varepsilon;x)$. $x_1$ is optimal when $\varepsilon=0$, but $x_2$ is optimal when $\varepsilon$ is sufficiently large. Although the nominal cost under $x_2$ is higher, average sensitivity is lower, ${\mathcal S}(\varepsilon;x_2)< {\mathcal S}(\varepsilon;x_1)$.
}
\label{fig:sensitivity-plots3}
\end{figure}

Suppose that $V(\varepsilon)$ is right differentiable at $\varepsilon=0$. By \eqref{eq:sensitivity-general2}, average sensitivity equals worst-case sensitivity in the limit as $\varepsilon \downarrow 0$
\begin{align*}
{\mathcal S}_{\mathbb P}\big[f(x, Y)\big] & = \lim_{\varepsilon\downarrow 0}{\mathcal S}\big(\varepsilon; f(x, Y)\big) = \lim_{\varepsilon\downarrow 0}\frac{1}{\varepsilon}{\mathcal A}\big(\varepsilon; f(x, Y)\big).
\end{align*}
In particular, DRO is a tradeoff between mean and worst-case sensitivity when $\varepsilon$ is small
\begin{align}
\min_x V(\varepsilon,\,x) & \equiv \min_x {\mathbb E}_{\mathbb P}[f(x, Y)] + \varepsilon {\mathcal S}_{\mathbb P}\big[f(x, Y)\big] + o(\varepsilon).
\label{eq:vf-theta1}
\end{align}


In many situations, ${\mathcal A}\big(\varepsilon; f(x, Y)\big)$ is also concave  in $\varepsilon$ in addition to being increasing and non-negative. This is true for all models considered in this paper\footnote{$V(\varepsilon; f(x, \cdot))$, and hence   ${\mathcal A}\big(\varepsilon; f(x, Y)\big)$, is concave in $\varepsilon$ if the set of alternative probability measures can be written in the form ${\mathcal Q}(\varepsilon) = \{{\mathbb Q}\,|\, d({\mathbb Q})\leq \varepsilon\}$ for some convex function $d$ of ${\mathbb Q}$.}. In this case, average sensitivity ${\mathcal S}\big(\varepsilon;f(x, Y)\big)$ is decreasing in $\varepsilon$ and worst-case sensitivity is a tight upper bound:
\begin{align*}
{\mathcal S}_{\mathbb P}\big[f(x, Y)\big] &\geq {\mathcal S}\big(\varepsilon;f(x, Y)\big)
\quad\mbox{for all }\varepsilon \geq 0.
\end{align*}

\subsection*{Sensitivity is a measure of deviation}
A {\it Generalized Measure of Deviation} \cite{rockafellar2006generalized} measures the spread of a random variable, generalizing the notion of the standard deviation.
\begin{definition}\label{def:deviation}
Let $f$ be a  random variable.
 ${\mathcal H}[f]$ is a {\it Generalized Measure of Deviation} or a {\it Generalized Measure of Spread} of $f$ if
\begin{enumerate}
\item ${\mathcal H}[f] \geq 0$ with equality if and only if $f$ is constant;
\item ${\mathcal H}[\beta f] = \beta{\mathcal H}[f]$ for every constant $\beta\geq 0$;
\item ${\mathcal H}[\alpha + f]={\mathcal H}[f]$ for every constant $\alpha\in{\mathbb R}$.
\end{enumerate}

For the rest of the paper, we focus on discrete random variables, which allows us to treat a random variable as an $n$-vector.
Let $f_i:=f(Y_i)$, and denote $\mathsf{f}:=(f_1,...,f_n)^\top$ 
and $\mathsf{p}:=(p_1,...,p_n)^\top$, where 
$p_i$ is the probability mass on $f_i$.
Without loss of generality, we assume that $p_i>0$ for all $i=1,...,n$. 
Likewise, we reserve $\mathsf{q}:=(q_1,...,q_n)^\top$ for an alternative probability distribution $\mathbb{Q}$.
 Furthermore, $\mathbb{E}_{\mathsf{p}}(\mathsf{f}):=\mathsf{p}^\top\mathsf{f}$, $\mathbb{V}_{\mathsf{p}}(\mathsf{f})$, and $\mathcal{S}_{\mathsf{p}}(\mathsf{f})$ denote, respectively, expectation $\mathbb{E}_{\mathbb{P}}[f]$, variance $\mathbb{V}_{\mathbb{P}}[f]$, and worst-case sensitivity $\mathcal{S}_{\mathbb{P}}[f]$ of $\mathsf{f}\equiv f$ under $\mathsf{p}\equiv \mathbb{P}$. Many of our results have generalizations to more complex settings, though the intuition from the discrete setting carries across. The restriction to discrete random variables enables us to communicate our message with minimal technical fuss.

\end{definition}
The following result shows that worst-case sensitivity is a {\it generalized measure of deviation}, and hence, measures the spread of the cost distribution under the nominal.  It is well known that ${\mathcal A}(\varepsilon;f)$ and ${\mathcal S}(\varepsilon;f)$ are generalized measures of deviation, which we include in the statement for completeness.  The proof is in the Appendix.
\begin{proposition}
\label{prop:deviation}
Let $\mathsf{f}\in\mathbb{R}^n$ denote the support of a random variable $f$, and let $\mathsf{p}\in\mathbb{R}^n$ satisfy $\mathsf{1}^\top\mathsf{p}=1,\mathsf{p}\geq\mathsf{0}$.
Suppose that $d(\mathsf{q}|\mathsf{p}):\mathbb{R}^{n}\times\mathbb{R}^{n}\to\mathbb{R}$ is convex and continuous in $\mathsf{q}$ and $d(\mathsf{q}|\mathsf{p})=0$ if $\mathsf{q}=\mathsf{p}$. Then the ambiguity cost \eqref{eq:DRO mean ambiguity-cost} satisfies
\begin{align*}
{\mathcal A}(\varepsilon;f) & = \max_{\mathsf{q}\in{\mathcal Q}(\varepsilon)}\sum_{i=1}^n q_i \Big(f_i-{\mathbb E}_{\mathbb P}[f
]\Big), \\
{\mathcal Q}(\varepsilon)   & = \Big\{\mathsf{q}=(q_1,\cdots,q_n)^\top\in\mathbb{R}^n~\Big|~\mathsf{1}^\top\mathsf{q}=1,~\mathsf{q\geq 0},~d(\mathsf{q}|\mathsf{p})\leq\varepsilon\Big\}.
\end{align*}
 For every $\varepsilon>0$, the ambiguity cost ${\mathcal A}(\varepsilon;f)$ and average sensitivity ${\mathcal S}(\varepsilon;f)=\frac{1}{\varepsilon}{\mathcal A}(\varepsilon;f)$ are generalized measures of deviation of $f$.
If  there is a constant $k>0$ such that
\begin{align}
d(\mathsf{p}+\delta \mathsf{\Delta}|\mathsf{p}) & \sim O(\delta^k) ~\mbox{when}  ~\delta \rightarrow 0,
\label{eq:continuity d}
\end{align}
for every $\mathsf{\Delta}\in{\mathbb R}^n$ such that ${\mathsf 1}'\mathsf{\Delta}=0$,
then ${\mathcal A}(\varepsilon;f)\sim O(\varepsilon^\frac{1}{k})$ and  worst-case sensitivity ${\mathcal S}_{\mathbb P}[f(x, Y)]$ defined by \eqref{eq:sensitivity-general2}  is a generalized measure of deviation with $g(\varepsilon)=\varepsilon^\frac{1}{k}$.
\end{proposition}
When the uncertainty set is a constraint on smooth $\phi$-divergence, $g(\varepsilon)=\sqrt{\varepsilon}$ in the definition \eqref{eq:sensitivity-general2} of worst-case sensitivity and \eqref{eq:continuity d} holds with $k=2$. It is linear in $\varepsilon$ in all other cases considered in this paper. Proposition \ref{prop:deviation} shows that  $g(\varepsilon)$ is determined by the continuity property \eqref{eq:continuity d} of the uncertainty set.

\section{Worst-case sensitivity: Explicit formulas}
\label{sec:WCS}
We derive explicit expressions for worst-case sensitivity for uncertainty sets associated with smooth $\phi$-divergence, Total Variation, budgeted uncertainty sets, uncertainty sets corresponding to a convex combination of the nominal distribution and a CVaR-type uncertainty set, and the Wasserstein metric. For the purposes of readability, all proofs can be found in the Appendix.

{

\subsection{Smooth $\phi$-divergences}
\label{sec:phi-div}
Consider the worst-case objective
\begin{align}
V_{\phi}(\varepsilon;\mathsf{f}) &:= \max_{\mathsf{q}\in{\mathcal Q}_{\phi}(\varepsilon)}\sum_{i=1}^{n}q_i 
f_i,
\label{eq:phi-theta}
\end{align}
where
\begin{align*}
{\mathcal Q}_{\phi}(\varepsilon) &:= \Big\{\mathsf{q}:=(q_1,...,q_n)^\top\in{\mathbb R}^n \, \Big| \, \sum_{i=1}^{n}p_i \phi\Big(\frac{q_i}{p_i}\Big) \leq \varepsilon,~\mathsf{1}^\top\mathsf{q}=1,~\mathsf{q}\geq\mathsf{0}\Big\}.
\end{align*}
We assume the following.
\begin{assumption}
\label{ass:phi}
$\phi(z)$ is strictly convex, twice continuously differentiable in $z$, with $\phi(1)=0$, $\phi'(1)=0$ and $\phi''(1)>0$.
\end{assumption}

Since $p_i>0$ for all $i$, Assumption \ref{ass:phi} and
convex duality imply that
for small $\varepsilon>0$,
\begin{align*}
V_{\phi}(\varepsilon) &= \min_{\delta>0,\, c} \max_{\mathsf{q>0}} \sum_{i=1}^{n}q_i f
 + \frac{1}{\delta}\Big(\varepsilon-\sum_{i=1}^{n}p_i \phi\Big(\frac{q_i}{p_i}\Big)\Big)  + c\Big(\sum_{i=1}^{n}q_i - 1\Big).
\end{align*}
The following result characterizes the solution $(c(\varepsilon), \delta(\varepsilon))$ of the dual problem and the associated worst-case distribution ${\mathsf q}(\varepsilon)=\big(q_1(\varepsilon),\cdots,\,q_n(\varepsilon)\big)^\top$ when $\varepsilon$ is small.
\begin{proposition}
\label{prop:phi-div}
Suppose that $\phi$ satisfies Assumption \ref{ass:phi}. Then
\begin{align}
c(\varepsilon)  & =  -{\mathbb E}_{\mathsf{p}}(\mathsf{f}) + O(\sqrt{\varepsilon}), \label{eq:c_phi}
\\
\delta(\varepsilon) & =  \sqrt{\varepsilon}\sqrt{\frac{2 \phi''(1)}{{\mathbb V}_{\mathsf{p}}(\mathsf{f})}}+ o(\sqrt{\varepsilon}).
\label{eq:delta_phi}
\end{align}
The family of worst-case distributions $\{{\mathsf q}(\varepsilon)\,|\,\varepsilon\geq 0\}$ satisfies
\begin{align}
q_i(\varepsilon)
& = p_i\Big\{1 + \sqrt\frac{2\varepsilon}{{{\mathbb V}_{\mathsf{p}}(\mathsf{f})}} \left(f_i-\mathbb{E}_{\mathsf{p}}(\mathsf{f})\right)\Big\}+ o(\sqrt\varepsilon).
\label{eq:wcq}
\end{align}
\end{proposition}

The worst-case expected cost is the expected cost under the worst-case distribution
\begin{align}
V_{\phi}(\varepsilon) & = \mathbb{E}_{\mathsf{q}(\varepsilon)}(\mathsf{f})\nonumber \\
& = 
\mathbb{E}_{\mathsf{p}}(\mathsf{f}) +\sqrt{\varepsilon}\sqrt{\frac{2\mathbb{V}_{\mathsf{p}}(\mathsf{f})}{\phi''(1)}} + o(\sqrt{\varepsilon}).
\label{eq:vf-exp-phi}
\end{align}
It follows that $V_{\phi}(\varepsilon)-V_{\phi}(0)$ is $O(\sqrt \varepsilon)$, so worst-case sensitivity \eqref{eq:sensitivity-general2} with $g(\varepsilon)=\sqrt \varepsilon$ is as follows.
\begin{proposition}
Suppose Assumption \ref{ass:phi} is satisfied. Then
\begin{align}
{\mathcal S}_{\mathsf{p}}(\mathsf{f})
 &=\lim_{\varepsilon\downarrow 0}\frac{V_{\phi}(\varepsilon)-V_{\phi}(0)}{\sqrt{\varepsilon}}
  =\sqrt{\frac{2{\mathbb V}_{\mathsf{p}}(\mathsf{f})}{\phi''(1)}}.
\label{eq:phi-sensitivity}
\end{align}
\end{proposition}
A closely related result for the case $\phi(z)$ is relative entropy was derived in \cite{lam2016robust}, while \cite{gotoh2018robust} derives worst-case sensitivity for the ``penalty formulation" of the DRO model.

\begin{example}
When $\phi$-divergence is modified modified $\chi^2$, $\phi(z) = \frac{1}{2}(z-1)^2$
\begin{align*}
c(\varepsilon) &=  -\mathbb{E}_{\mathsf{p}}(\mathsf{f}),\quad
\delta(\varepsilon)  =  \sqrt\frac{2\varepsilon}{{\mathbb{V}_{\mathsf{p}}(\mathsf{f})}}
\end{align*}
and the worst-case distribution is
\begin{align*}
q_i(\varepsilon) & =  p_i\big\{1 + \delta(f_i + c )\big\}= p_i\Big\{1 + \sqrt\frac{2\varepsilon}{{{\mathbb V}_{\mathsf{p}}(\mathsf{f})}}\Big(f_i-\mathbb{E}_{\mathsf{p}}(\mathsf{f})\Big)\Big\}.
\end{align*}
This holds for all $\varepsilon\geq 0$ as long as $q_i(\varepsilon)\geq 0$, and not just when it is small.
It follows that
\begin{align*}
V_{\phi}(\varepsilon) &= \sum_{i=1}^np_i f_i + \sqrt{\varepsilon}\sqrt{2 {\mathbb V}_{\mathsf{p}}(\mathsf{f})}.
\end{align*}
Clearly, $V_{\phi}(\varepsilon) - V_{\phi}(0) \sim O(\sqrt \varepsilon)$ and worst-case sensitivity is
\begin{align}
{\mathcal S}_{\mathsf{p}}(\mathsf{f}) &= \sqrt{2 {\mathbb V}_{\mathsf{p}}(\mathsf{f})}.
\label{eq:wcs-chi2}
\end{align}
\end{example}

\begin{example} \label{example:GKL}
In \cite{gotoh2018robust}, the penalty version of the worst-case problem is used to define worst-case sensitivity. Specifically, a family of worst-case distributions $\{\tilde{\mathsf q}(\delta)\,|\,\delta\geq 0\}$ is given by the solutions of the worst-case problem
\begin{align}
\tilde{\mathsf{q}}(\delta) & := \left\{
\begin{array}{cl}
\begin{displaystyle}\argmax_{\mathsf q}\Big\{  \sum_{i=1}^nq_i f_i -  \frac{1}{\delta} \sum_{i=1}^n {p}_i \phi\Big(\frac{q_i}{{p}_i}\Big)\Big\}, \end{displaystyle}& \delta>0,\\
\mathsf{p}, & \delta=0,
\end{array}\right.
\label{eq:worst-case-Q}
\end{align}
where the parameter $\delta$ determines the penalty on deviations from the nominal. In particular, $\delta=0$ gives the nominal and increasing $\delta$ is analogous to increasing the size of the uncertainty set in \eqref{eq:phi-theta}. When the penalty version is used to define the set of worst-case measures, the worst-case expected cost under 
$\tilde{\mathsf{q}}(\delta)$ is linear in the ambiguity parameter
\begin{align*}
V(\delta) & \equiv \mathbb{E}_{\tilde{\mathsf{q}}(\delta)}(\mathsf{f}) = \mathbb{E}_{\mathsf{p}}(\mathsf{f})
                +\frac{\delta}{\phi''(1)}\mathbb{V}_{\mathsf{p}}(\mathsf{f}) + o(\delta),
\end{align*}
so the standard definition of sensitivity \eqref{eq:sensitivity-general1} can be used:
\begin{align}
{\mathcal S}_{\mathsf{p}}(\mathsf{f})
 &=\lim_{\delta\downarrow 0}\frac{\mathbb{E}_{\tilde{\mathsf{q}}(\delta)}(\mathsf{f}) -\mathbb{E}_{\mathsf{p}}(\mathsf{f})}{\delta}
  =\frac{1}{\phi''(1)}\mathbb{V}_{\mathsf{p}}(\mathsf{f}).
\label{eq:phi-sensitivity-var}
\end{align}
While this leads to a different sensitivity measure, the qualitative nature is the same as \eqref{eq:phi-sensitivity}.
\end{example}

\subsection{
Total Variation}
Consider the worst-case expected cost
\[
V_{\rm TV}(\varepsilon;\mathsf{f}) := \max_{\mathsf{q}\in\mathcal{Q}_\mathrm{TV}(\varepsilon)}~\mathbb{E}_{\mathsf{q}}(\mathsf{f})
\]
with uncertainty set
\[
\mathcal{Q}_\mathrm{TV}(\varepsilon):=\Big\{\mathsf{q}\in\mathbb{R}^{n}\,\Big|\,
\mathsf{1}^\top|\mathsf{q}-\mathsf{p}|\leq\varepsilon,~
\mathsf{1}^\top\mathsf{q}=1,~\mathsf{q}\geq\mathsf{0}\Big\},
\]
for $\varepsilon\geq 0$, where $|\mathsf{z}|:=(|z_1|,...,|z_n|)^\top$. (We can focus on $\varepsilon\leq 2$ since the set coincides with the unit simplex $\{\mathsf{q}\in\mathbb{R}^{n}|\mathsf{1}^\top\mathsf{q}=1,~\mathsf{q}\geq\mathsf{0}\}$
  ~otherwise.) This uncertainty set is equivalent to a constraint on $\phi$-divergence with $\phi(z)=|z-1|$.
Note however that $\phi(z)$ it is not differentiable at $z=1$ so the results from Section  \ref{sec:phi-div} do not apply.

Consider an ordering of the components of the cost vector $\mathsf{f}=(f_1,...,f_n)^\top\in\mathbb{R}^n$ from largest to smallest and denote the  $i^{th}$ largest component by $f_{(i)}$, i.e.,
\begin{equation}
f_{(1)}\geq\cdots\geq f_{(n)},
\label{eq:ordered_f}
\end{equation}
and let $p_{(i)}$ denote the probability mass corresponding to $f_{(i)}$.
The following lemma characterizes the worst-case objective for sufficiently small $\varepsilon$.
\begin{lemma}\label{lemma:sen_tv}
Suppose that $f$ corresponds to a nonconstant random variable and $\varepsilon \in (0, \min(\mathsf{p}))$.
Then a worst-case probability distribution is
\begin{equation}
(q_{(1)},q_{(2)},...,q_{(n-1)},q_{(n)})=
\big(p_{(1)}+\frac{\varepsilon}{2},~p_{(2)},~...,~p_{(n-1)},~p_{(n)}-\frac{\varepsilon}{2}\big)
\label{eq:worst-case_solution_to_tv}
\end{equation}
and the worst-case objective is
\[
V_{\rm TV}(\varepsilon;\mathsf{f})=
\mathbb{E}_\mathsf{p}(\mathsf{f})+\frac{\varepsilon(\max(\mathsf{f})-\min(\mathsf{f}))}{2},
\]
where $q_{(i)}$ denotes the worst-case probability mass corresponding to $f_{(i)}$.
\end{lemma}
The expression for worst-case sensitivity follows immediately.
\begin{proposition}
\label{prop:sen_tv}
For the Total Variation uncertainty set $\mathcal{Q}_\mathrm{TV}(\varepsilon)$, worst-case sensitivity
\begin{align}
{\mathcal S}_{\mathsf{p}}(\mathsf{f})
= \frac{\max(\mathsf{f})-\min(\mathsf{f})}{2}
\equiv\frac{1}{2}\times\mbox{``Range of }\mathsf{f}.\mbox{''}
\label{eq:sen_tv}
\end{align}
\end{proposition}

It is known (e.g. \cite{shiffler1980}) that for any $\mathsf{f}$ and $\mathsf{p}=\mathsf{1}/n$,
\[
\sqrt{\mathbb{V}_{\mathsf{p}}(\mathsf{f})}\leq \frac{1}{2}\mbox{Range}(\mathsf{f}),
\]
and, accordingly, we have
\begin{eqnarray}
\sqrt{\frac{\phi''(1)}{2}}V'_{\phi}(0^+)\leq V'_{{\rm TV}}(0^+).
\label{eq:phi-TV bound}
\end{eqnarray}
This suggests that when $\varepsilon$ is small, the solution of the DRO problem with the Total Variation uncertainty set will be close to optimal for DRO  with smooth $\phi$-divergence.

\subsection{Budgeted uncertainty}
\label{sec:CVaR}
Consider the uncertainty set
\begin{align}
{\mathcal Q}_{\mathrm{b}}(\varepsilon) & = \Big\{\mathsf{q}\in\mathbb{R}^n\,\Big|\,\mathsf{1}^\top\mathsf{q}=1,\mathsf{0}\leq\mathsf{q}\leq(1+\varepsilon)\mathsf{p}\Big\}
\label{eq:CVaR_uncertainty}
\end{align}
and worst-case expected cost
\begin{equation}
V_{\rm 
b}(\varepsilon;\mathsf{f}):=\max_{\mathsf{q}\in{\mathcal Q}_{\rm b}(\varepsilon)} \mathbb{E}_{\mathsf{q}}(\mathsf{f}),
\label{eq:dual_cvar_eps}
\end{equation}
for $\varepsilon\geq 0$. (We can focus on $\varepsilon\leq \max_i\{\frac{1}{p_i}-1\}$ since ${\mathcal Q}_{\rm b}(\varepsilon)$ is a set of probability distributions ~otherwise.)
 For $\varepsilon\in(0,\min_i\{\frac{1}{p_{i}}-1\})$, the worst-case distribution is given by
\[
(q_{(1)},..., q_{(k)},q_{(k+1)},q_{(k+2)},...,q_{(n)})=
\big((1+\varepsilon)p_{(1)},...,(1+\varepsilon)p_{(k)},1-(1+\varepsilon)\sum_{i=1}^kp_{(i)},0,...,0\big).
\]
The set \eqref{eq:CVaR_uncertainty} can be referred to as the \emph{budgeted uncertainty set} and is related to the {\it Conditional Value-at-Risk} with parameter $\alpha\in(0, 1)$ ($\alpha$-CVaR) 
\begin{equation}
\mathrm{CVaR}_{\mathsf{p},\alpha}(\mathsf{f})=\max_{\mathsf{q}}\Big\{
\mathsf{f}^\top\mathsf{q}
\,\Big|\,\mathsf{1}^\top\mathsf{q}=1,\mathsf{0}\leq\mathsf{q}\leq\frac{1}{1-\alpha}\mathsf{p}\Big\}
\label{eq:dual_cvar}
\end{equation}
Obviously, $V_{\rm b}(\varepsilon)=\mathrm{CVaR}_{\mathsf{p},\frac{\varepsilon}{1+\varepsilon}}(\mathsf{f})$.

$V_{\rm b}(\varepsilon
)$ is piecewise linear, concave, and increasing in $\varepsilon$. 
The following result characterizes the slope of the worst-case expected cost $V_{\rm 
b}(\varepsilon)$ for all values of $\varepsilon$.
\begin{proposition}
\label{prop:sen_cvar_env}
Let $\varepsilon>0$. Suppose $k\in \{1,...,n\}$ is an integer such that
\begin{equation}
\varepsilon \in \Big[\frac{\sum_{i=k+1}^{n}p_{(i)}}{\sum_{i=1}^{k}p_{(i)}},\frac{\sum_{i=k}^{n}p_{(i)}}{\sum_{i=1}^{k-1}p_{(i)}}\Big),
\label{eq:intervals_cvar}
\end{equation}
where $p_{(i)}$ is the probability mass of the $i$-th largest cost, $f_{(i)}$, and $\sum\limits_{i=n+1}^np_{(i)}=\sum\limits_{i=1}^0p_{(i)}=0$ and $1/0=\infty$.
For 
$\Delta>0$ satisfying $\varepsilon+\Delta<\frac{\sum_{i=k}^{n}p_{(i)}}{\sum_{i=1}^{k-1}p_{(i)}}$, we have
\begin{align}
S_{\rm b}(\varepsilon):=
\frac{V_{\rm b}(\varepsilon+\Delta
)-V_{\rm b}(\varepsilon
)}{\Delta}
&=\sum_{i=1}^kp_{(i)}\big(f_{(i)}-f_{(k+1)}\big)\label{eq:constant_slope_cvar}\\
&=\frac{1}{1+\varepsilon}\Big(
\mathrm{CVaR}_{\mathsf{p},\frac{\varepsilon}{1+\varepsilon}}(\mathsf{f})-\mathrm{VaR}_{\mathsf{p},\frac{\varepsilon}{1+\varepsilon}}(\mathsf{f})
\Big).
\label{eq:sen_cvar_env}
\end{align}
\end{proposition}
\eqref{eq:constant_slope_cvar} defines the constant slope of the linear piece over the interval \eqref{eq:intervals_cvar}.
Let
\begin{eqnarray*}\varepsilon_{(h)}=\frac{\sum_{i=n-h+1}^{n}p_{(i)}}{\sum_{i=1}^{n-h}p_{(i)}},h=1,2,...,n.
\end{eqnarray*}
 For $\varepsilon\in[\varepsilon_{(k)},\varepsilon_{(k+1)})$,
\begin{align*}
V_{\rm b}(\varepsilon)
&=V_{\rm b}(0)+\sum_{h=0}^{k-1}S_{\rm b}(\varepsilon_h)(\varepsilon_{(h+1)}-\varepsilon_{(h)})+S_{\rm b}(\varepsilon_k)(\varepsilon-\varepsilon_{(k)})\\
&=V_{\rm b}(0)+\sum_{h=0}^{k-1}\frac{p_{(h)}}{1+\varepsilon_{(h)}}\big(\mathrm{CVaR}_{\mathsf{p},\frac{\varepsilon_{(h)}}{1+\varepsilon_{(h)}}}-\mathrm{VaR}_{\mathsf{p},\frac{\varepsilon_{(h)}}{1+\varepsilon_{(h)}}}\big)(\varepsilon_{(h+1)}-\varepsilon_{(h)})\\
&\qquad\qquad\qquad +\frac{p_{(k)}}{1+\varepsilon_{(k)}}\big(\mathrm{CVaR}_{\mathsf{p},\frac{\varepsilon_{(k)}}{1+\varepsilon_{(k)}}}-\mathrm{VaR}_{\mathsf{p},\frac{\varepsilon_{(k)}}{1+\varepsilon_{(k)}}}\big)(\varepsilon-\varepsilon_{(k)}).
\end{align*}

%
Since $V_{\rm b}(\varepsilon
)$ is concave and increasing,  
its slope is the largest over the left-most piece $(0,\frac{p_{(n)}}{1-p_{(n)}})$. For 
$\varepsilon\in(0,\frac{p_{(n)}}{1-p_{(n)}})$ and $\Delta$ sufficiently small, \eqref{eq:sen_cvar_env} becomes
\begin{align*}
\frac{V_{\rm b}(\varepsilon+\Delta
)-V_{\rm b}(\varepsilon
)}{\Delta} & = {\mathbb E}_{\mathsf{p}}(\mathsf{f}) - \min(\mathsf{f}).
\end{align*}
The following expression for worst-case sensitivity follows immediately.
\begin{corollary}
\label{cor:sen_cvar}
 For \eqref{eq:dual_cvar_eps}, we have
\begin{align}
{\mathcal S}_{\mathsf{p}}(\mathsf{f}) = \mathbb{E}_{\mathsf{p}}(\mathsf{f})-\min(\mathsf{f}).
\label{eq:sen_cvar}
\end{align}
\end{corollary}


While worst-case sensitivity \eqref{eq:sen_cvar} is a measure of spread, it only depends on the  ``good" side of the cost distribution. In contrast, smooth $\phi$-divergence \eqref{eq:phi-sensitivity} and Total Variation \eqref{eq:sen_tv} depend on the entire distribution. We now see an uncertainty set where worst-case sensitivity depends on the spread of the ``bad" part of the cost-distribution.

\subsection{Convex Combination of Expected Loss and CVaR}
\label{sec:ExpLoss_CVaR}

Let $\mathsf{p}$ be the nominal distribution, $\alpha\in[0,1)$ be a fixed parameter, and consider the uncertainty set
\begin{align}
\mathcal{Q}_{\rm c}(\varepsilon)&:=
(1-\varepsilon)\{\mathsf{p}\}+\varepsilon\mathcal{Q}_{\rm CVaR}(\alpha)
\label{eq:CVaR_uncertainty_set}
\end{align}
parameterized by $\varepsilon\in[0, 1]$ where
\begin{align*}
\mathcal{Q}_{\rm CVaR}(\alpha):=
\Big\{\mathsf{q}\in\mathbb{R}^{n}\,\Big|\,
\mathsf{1}^\top\mathsf{q}=1,~\mathsf{0}\leq\mathsf{q}\leq\frac{1}{1-\alpha}\mathsf{p}\Big\}
\end{align*}
is the 
feasible set of \eqref{eq:dual_cvar}.
The worst-case expected cost is
\begin{align}
V_{\rm c}(\varepsilon;\mathsf{f})&:=\max_{\mathsf{q}\in\mathcal{Q}_\mathrm{c}(\varepsilon)}~\mathbb{E}_{\mathsf{q}}(\mathsf{f}).
\label{eq:dual_comb_eps}
\end{align}
Observe that $\mathcal{Q}_{\rm c}(0)=\{\mathsf{p}\}$, so there is no robustness if $\varepsilon=0$ and the worst-case expected cost is SAA.
The uncertainty set \eqref{eq:CVaR_uncertainty_set}  was considered in \cite{anderson2019robust} and is equivalent to
\begin{align*}
\mathcal{Q}_{\rm c}(\varepsilon) &:= \Big\{\mathsf{q}\in{\mathbb R}^n\,\Big|\, \mathsf{1}^\top\mathsf{q}=1,~\mathsf{p}(1-\varepsilon)\leq \mathsf{q}\leq \mathsf{p}(1-\varepsilon)+\frac{\varepsilon}{1-\alpha}\mathsf{p}\Big\}.
\end{align*}

Adopting the convention \eqref{eq:ordered_f}, the worst-case probability distribution is given by
\begin{align*}
\lefteqn{\big(q_{(1)},~...,~q_{(k)},~q_{(k+1)},~q_{(k+2)},~...,~q_{(n)}\big) \equiv  \argmax_{\mathsf{q}\in\mathcal{Q}_\mathrm{c}(\varepsilon)}~\mathbb{E}_{\mathsf{q}}(\mathsf{f})} \\
&=(1-\varepsilon)\big(p_{(1)},~...,~p_{(k)},~p_{(k+1)},~p_{(k+2)},~...,~p_{(n)}\big)\\
&\qquad +\varepsilon\big(\frac{1}{1-\alpha}p_{(1)},...,\frac{1}{1-\alpha}p_{(k)},1-\frac{1}{1-\alpha}\sum_{i=1}^kp_{(i)},0,...,0\big)\\
&=\Big((1+\frac{\alpha\varepsilon}{1-\alpha})p_{(1)},...,(1+\frac{\alpha\varepsilon}{1-\alpha})p_{(k)},(1-\varepsilon)p_{(k+1)}+\varepsilon\big(1-\frac{1}{1-\alpha}\sum_{i=1}^kp_{(i)}\big),\\
&\qquad\qquad\qquad\qquad(1-\varepsilon)p_{(k+2)},...,(1-\varepsilon)p_{(n)}\Big).
\end{align*}

It can be shown that the worst-case objective satisfies
\[
V_{{\rm c}}(\varepsilon;\mathsf{f})
:=(1-\varepsilon)\mathbb{E}_\mathsf{p}(\mathsf{f})+\varepsilon\mathrm{CVaR}_{\mathsf{p},\alpha}(\mathsf{f}).
\]
It 
 follows that for 
any non-uniform vector $\mathsf{f}$ and $\varepsilon\in(0,1]$, the function $V_{\rm c}(\varepsilon)$ is linearly increasing at a rate of its \emph{CVaR Deviation} \cite{rockafellar2006generalized} 
\begin{equation}
\frac{
V_{{\rm c}}(\varepsilon)-\mathbb{E}_\mathsf{p}(\mathsf{f})}{\varepsilon}
=\mathrm{CVaR}_{\mathsf{p},\alpha}(\mathsf{f})-\mathbb{E}_\mathsf{p}(\mathsf{f})
.
\label{eq:sen_cvxcmb_expectation+cvar}
\end{equation}
The following expression for worst-case sensitivity is obtained by letting $\varepsilon\searrow 0$.
\begin{corollary}\label{cor:wcs:comb}
For \eqref{eq:dual_comb_eps}, we have
\begin{align}
{\mathcal S}_{\mathsf{p}}(\mathsf{f}) &= \mathrm{CVaR}_{\mathsf{p},\alpha}(\mathsf{f})-\mathbb{E}_\mathsf{p}(\mathsf{f})
~\equiv~``\mbox{CVaR Deviation of }\mathsf{f}.\mbox{''}
\label{eq:sen-CVaRdev}
\end{align}
\end{corollary}

Since $\mathrm{CVaR}_{\mathsf{p},\alpha}(\mathsf{f})=\max(\mathsf{f}):=\max\{f_1,...,f_n\}$ for $\alpha\in[1-p_{(1)},1)$, worst-case sensitivity is the spread of the  ``bad" part of the cost distribution, $\mathcal{S}_{\mathsf{p}}(\mathsf{f})=\max(\mathsf{f})-\mathbb{E}_\mathsf{p}(\mathsf{f})$, which contrasts with \eqref{eq:sen_cvar}.
\begin{remark}
If $0 \leq L\leq 1 \leq U$, the uncertainty set
\begin{align*}
\mathcal{Q}_{\rm w}(L,U) & := \Big\{\mathsf{q}\in{\mathbb R}^n\,\Big|\,\mathsf{1}^\top\mathsf{q}=1,~
L\mathsf{p}\leq\mathsf{q}\leq U\mathsf{p}
\Big\}
\end{align*}
is equivalent to $\phi$-divergence with $\phi(z)=\delta_{[L,U]}(z)$. The worst-case expected cost is
\begin{align*}
V_{\rm w}(L,U;\mathsf{f})&:=
L\cdot\mathbb{E}_\mathsf{p}(\mathsf{f})+(1-L)\cdot\mathrm{CVaR}_{\mathsf{p},\frac{U-1}{U-L}}(\mathsf{f}).
\end{align*}
If $(L,U)=(0,\frac{1}{1-\alpha})$, $V_{\rm w}(L,U)$ is $\alpha$-CVaR.
If $(L,U)=(1-\varepsilon,\frac{1-(1-\varepsilon)\alpha}{1-\alpha})$, $V_{\rm w}(L,U)$ is $V_{\rm c}(\varepsilon)$.
While the parameter $\alpha$ is usually fixed (e.g., at $0.95$ or $0.99$), it can be viewed as another hyperparameter, in addition to $\varepsilon$, that defines the uncertainty set. A reasonable option to merge the two parameters into a single one is to set as $U=\frac{1}{L}=1+\nu>0$,
and the uncertainty set becomes
\begin{align*}
\mathcal{Q}_{\rm s}(\nu)&:=\Big\{\mathsf{q}
\,\Big|\, \mathsf{1}^\top\mathsf{q}=1,
\frac{1}{1+\nu}\mathsf{p}\leq\mathsf{q}\leq(1+\nu)\mathsf{p}
\Big\}=\Big\{\mathsf{q}
\,\Big|\, \mathsf{1}^\top\mathsf{q}=1,
\frac{1}{1+\nu}\mathsf{q}\leq\mathsf{p}\leq(1+\nu)\mathsf{q}
\Big\}.
\end{align*}
It is easy to see that the worst-case sensitivity is then
$\mathcal{S}_{\mathsf{p}}(\mathsf{f})=\mathrm{CVaR}_{\mathsf{p},\frac{1}{2}}(\mathsf{f})-\mathbb{E}_\mathsf{p}(\mathsf{f})$.
\end{remark}

We can associate CVaR deviation with the standard deviation.
\begin{proposition}
\label{propo:M-Stdev_as_UB_of_CVaR}
Let $\mathsf{f}\in\mathbb{R}^n$ and $\alpha\in(0,1)$. 
For $\mathsf{p}=\mathsf{1}/n$,
we have
\begin{equation}
\mbox{``CVaR Deviation of }\mathsf{f}\mbox{''}~\equiv~
\mathrm{CVaR}_{\mathsf{p},\alpha}(\mathsf{f})-\mathbb{E}_{\mathsf{p}}(\mathsf{f})\leq C_{\alpha,n}\sqrt{\mathbb{V}_{\mathsf{p}}(\mathsf{f})},
\label{eq:M-Stdev_as_UB_of_CVaR}
\end{equation}
where 
\[
C_{\alpha,n}:=\frac{\sqrt{n\Big\{\lfloor\kappa\rfloor+\big(\kappa-\lfloor\kappa\rfloor\big)^2\Big\}-\kappa^2}}{\kappa}
\]
with $\kappa:=n(1-\alpha)$.
The inequality \eqref{eq:M-Stdev_as_UB_of_CVaR} is tight
, i.e., there is a 
vector $\mathsf{f}$ which attains the equality.
\end{proposition}
Note that $C_{\alpha,n}\leq\sqrt{\frac{\alpha}{1-\alpha}}$ for all $\alpha\in[0,1)$, and especially when $n(1-\alpha)\in\mathbb{Z}$, the equality holds. Accordingly, \eqref{eq:M-Stdev_as_UB_of_CVaR} suggests a relation between $\alpha$-CVaR and Mean-Standard Deviation:
\begin{align*}
\mathrm{CVaR}_{\mathsf{p},\alpha}(\mathsf{f})&\leq \mathbb{E}_{\mathsf{p}}(\mathsf{f})+\sqrt{\frac{\alpha}{1-\alpha}}\sqrt{\mathbb{V}_{\mathsf{p}}(\mathsf{f})},
\end{align*}
or
\begin{align*}
V_{\rm b}(\varepsilon;\mathsf{f})&\leq \mathbb{E}_{\mathsf{p}}(\mathsf{f})+\sqrt{\varepsilon}\sqrt{\mathbb{V}_{\mathsf{p}}(\mathsf{f})}.
\end{align*}
\begin{remark}
While Proposition 1 of \cite{rockafellar2014superquantile} shows a similar bound for random variables in the $L^2$-space, their coefficient is $1/\sqrt{1-\alpha}$, which is larger than $C_{\alpha,n}$.
\begin{figure}[h]\centering
\includegraphics[scale=0.75]{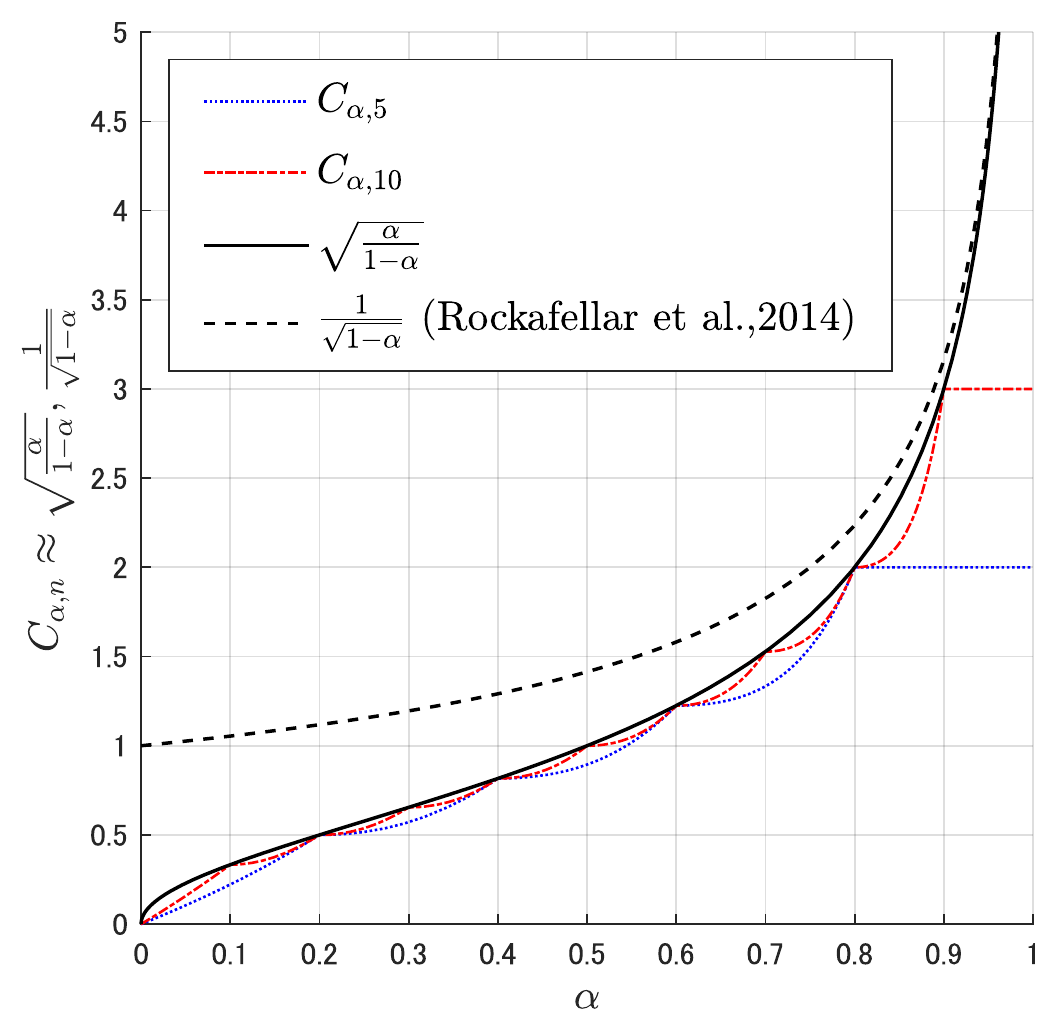}
\caption{$C_{\alpha,n}$, $\sqrt{\frac{\alpha}{1-\alpha}}$, and $\frac{1}{\sqrt{1-\alpha}}$}
\end{figure}
\end{remark}

The inequality \eqref{eq:M-Stdev_as_UB_of_CVaR} is applicable  to the worst-case sensitivity results.
First of all,
\begin{eqnarray}
V'_{{\rm c}}(0^+)\leq C_{\alpha,n}\sqrt{\frac{\phi''(1)}{2}}V'_{\phi}(0^+) \leq \sqrt{\frac{\phi''(1)\alpha}{2(1-\alpha)}}V'_{\phi}(0^+)
\label{eq:convex-phi bound}
\end{eqnarray}
where $V'_{\phi}(0^+)$ is the worst sensitivity of the DRO objective with the smooth $\phi$ divergence.
This inequality suggests that the sensitivity of the DRO with the convex combination of mean and CVaR is bounded above by that with any smooth $\phi$ divergence.
Second, recalling Corollary \ref{cor:sen_cvar}, we have a tight bound:
\begin{align*}
V'_{\rm b}(0^+)=
\mathbb{E}_{\mathsf{p}}(\mathsf{f})-\min(\mathsf{f})
&=\frac{\mathrm{CVaR}_{\mathsf{p},\frac{\varepsilon}{1+\varepsilon}}(\mathsf{f})-\mathbb{E}_\mathsf{p}(\mathsf{f})}{\varepsilon} 
\\
&\leq \varepsilon^{-1/2}\sqrt{\mathbb{V}_{\mathsf{p}}(\mathsf{f})}
=\sqrt{\frac{\phi''(1)}{2\varepsilon}}V'_{\phi}(0^+)
\end{align*}
for $\mathsf{p}=\mathsf{1}/n$.
Since \eqref{eq:sen_cvar} holds for any $\varepsilon\in(0,\frac{p_{(n)}}{1-p_{(n)}})=(0,\frac{1}{n-1})$, taking $\varepsilon=\frac{1}{n-1}$, we have a (loose) bound
\[
V'_{\rm b}(0^+)<\sqrt{\frac{(n-1)\phi''(1)}{2}}V'_{\phi}(0^+).
\]
From this we see that when $n$ is large, the difference between the smooth $\phi$ and the 
budgeted uncertainty $\phi=\delta_{[0,1+\varepsilon]}$ can be large for small uncertainty sets.
In contrast, the bounds (4.13) and (4.26) relating the sensitivities $V'_{\rm TV}(0^+)$, $V'_{\rm c}(0^+)$ and $V'_{\phi}(0^+)$ are tight and independent of $n$.
The potentially large difference likely reflects the fact that the sensitivity of ``budgeted uncertainty" depends only on the lower part of the cost distribution whereas  $V'_{\phi}(0^+)$ depends on the entire distribution.
More generally, this suggests the possibility that solutions of DRO problems with ``budgeted uncertainty" may differ quite substantially from those for other uncertainty sets.


\subsection{Wasserstein metric}
\label{sec:Wasserstein}
Consider the worst-case expected cost with a constraint on the Wasserstein metric \cite{blanchet2019,esfahani2018data,gao2017wasserstein}:
\begin{align}
V_{\rm w}(\varepsilon) &:= \max_{\gamma \in {\mathcal X}} \Big\{ \int_z f(x,\,z) \Big(\sum_{i=1}^{n} \gamma_i(dz) \Big) \,\Big|\,
\sum_{i=1}^{n}\int_z \|z-Y_i\|_p\gamma_i(dz) \leq \varepsilon \Big\},
\label{eq:W2}
\end{align}
where
\begin{align*}
{\mathcal X} = \Big\{\gamma\;\Big|\; 
\int_z\gamma_i(dz) = p_i,\; i=1,\cdots,\,n 
,\; 
\gamma_i(dz) \geq 0 
\Big\}.
\end{align*}
We are thinking of $1\leq p \leq \infty$ for the norm in the Wasserstein metric. Note that the cost and constraint functionals
\begin{align*}
F(\gamma) & = \max_\gamma \int_z f(x,\,z) \Big(\sum_{i=1}^{n} \gamma_i(dz) \Big)\\
G(\gamma) & := \sum_{i=1}^{n}\int_z \|z-Y_i\|_p\gamma_i(dz)
\end{align*}
are linear in $\gamma$, so \eqref{eq:W2} is a convex optimization problem, and $V_{\rm w}(\varepsilon)$ is concave, increasing and differentiable in $\varepsilon$ almost everywhere \cite{
luenberger1997optimization}.

Since solutions of the dual problem of \eqref{eq:W2} are supergradients of the value function $V_{\rm w}(\varepsilon)$ \cite{luenberger1997optimization}, we study worst-case sensitivity $V_{\rm w}'(0^+)$ by studying dual solutions when $\varepsilon\downarrow 0$.

Let $\lambda\geq 0$ be the Lagrange multiplier for the Wasserstein constraint. The dual problem is
\begin{align}
\lefteqn{\min_{\lambda\geq 0} \max_{\gamma \in{\mathcal X}} \Big\{ \sum_{i=1}^{n} \int_z f(z) \gamma_i(dz)
+ \lambda\Big(\varepsilon- \sum_{i=1}^n\int_z \|z_i - Y_i\|_p\gamma_i(dz)\Big)} \nonumber  \\[5pt]
& = \min_{\lambda\geq 0} \max_{\gamma \in {\mathcal X}} \sum_{i=1}^{n}\int_{z_i}\Big[f(z_i)-\lambda \|z_i - Y_i\|_p \Big]\gamma_i(dz_i) + \lambda \varepsilon
\label{eq:D1}
\end{align}
where to ease notation, we drop the decision variable from the notation and write $f(z) \equiv f(x,\,z)$.
This can be written
\begin{align*}
\lefteqn{\min_{\lambda\geq 0} \Big\{\sum_{i=1}^{n} p_i \max_{z_i}\Big\{f(z_i) -\lambda \|z_i - Y_i\|_p\Big\} + \lambda\varepsilon\Big\}} \\
& = \sum_{i=1}^{n} p_i f(Y_i) + \min_{\lambda\geq 0} \Big\{\sum_{i=1}^{n} p_i \max_{z_i}\Big\{f(z_i) - f(Y_i) -\lambda \|z_i - Y_i\|_p\Big\} + \lambda\varepsilon \Big\}.
\end{align*}
Intuitively, for every given transportation cost $\lambda$, the inner maximization in \eqref{eq:D1} defines a worst-case measure that moves probability mass $p_i$ from $Y_i$ to
\begin{align*}
z_i^* &= \argmax_{z_i}\Big\{f(z_i) -\lambda \|z_i - Y_i\|_p\Big\}.
\end{align*}
If $\lambda(\varepsilon)$ is any solution of the dual problem at $\varepsilon$,  $\lambda(\varepsilon)$ is a super-gradient of $V_{\rm w}(\varepsilon)$ at $\varepsilon$ \cite{
luenberger1997optimization}, so concavity of $V_{\rm w}(\varepsilon)$ means that worst-case sensitivity $V_{\rm w}'(0^+)\leq \lambda(0)$. We compute the sensitivity at $\varepsilon=0$ by characterizing the solutions of the dual problem of \eqref{eq:W2} when $\varepsilon=0$, and showing that one of these actually equals the right derivative $V_{\rm w}'(0^+)$.

By Lemma \ref{lemma:duality_prop}, strong duality holds when $\varepsilon>0$. The following result shows that  strong duality also holds when $\varepsilon=0$, and characterizes the set of optimal dual variables.
\begin{proposition} \label{prop:eps0}
Assume that there exists constant $L$ such that $|f(z)-f(Y_i)|\leq L\|z-Y_i\|_p$ for every $z$ and $i=1,\cdots,\,n$. Then
\begin{align*}
\min_{\lambda\geq 0}\sum_{i=1}^k p_i\max_{z_i} \Big\{f(z_i)-f(Y_i)- \lambda\|z_i-Y_i\|_p\Big\}=0
\end{align*}
and strong duality holds when $\varepsilon=0$
\begin{align*}
V_{\rm w}(0) & = \sum_{i=1}^{n}p_i f(Y_i) + \min_{\lambda\geq 0}\sum_{i=1}^k p_i\max_{z_i} \Big\{f(z_i)-f(Y_i)- \lambda\|z_i-Y_i\|_p\Big\}  \\
             & = \sum_{i=1}^{n}p_i f(Y_i),
\end{align*}
and hence for all $\varepsilon\geq 0$.
The set of optimal solutions of the dual problem when $\varepsilon=0$  is
\begin{align*}
\nonumber
\lefteqn{\Big\{\lambda\,\big|\, \lambda \geq \max_{i=1,\cdots,\,n} \max_{z_i}\frac{f(z_i)-f(Y_i)}{\|z_i-Y_i\|_p}\Big\}}\\
&\quad = \argmin_{\lambda\geq 0} \sum_{i=1}^k p_i\max_{z_i} \Big\{f(z_i)-f(Y_i)- \lambda\|z_i-Y_i\|_p\Big\}.
\end{align*}
\end{proposition}

Proposition \ref{prop:eps0} implies that if
\begin{align*}
\lambda \geq \max_{i=1,\cdots,\,n} \max_{z_i}\frac{f(z_i)-f(Y_i)}{\|z_i-Y_i\|_p},
\end{align*}
then $\lambda$ is a supergradient of $V_{\rm w}$ at $\varepsilon=0$, and hence is an upper bound of the right derivative
\begin{align}
V_{\rm w}'(0^+) \leq \max_{i=1,\cdots,\,n} \max_{z_i}\frac{f(z_i)-f(Y_i)}{\|z_i-Y_i\|_p}.
\label{eq:Wass_temp1}
\end{align}
The following result shows that this inequality is actually an equality, so $V_{\rm w}'(0^+)$ is also a solution of the dual problem at $\varepsilon=0$. This allows us to identify the identify worst-case sensitivity with the lower bound of the set of dual solutions at $\varepsilon=0$. The proof can be found in the Appendix.
\begin{proposition}
\label{prop:Wass_sensitivity}
 For \eqref{eq:W2}, we have
\begin{align}
{\mathcal S}_{\mathbb{P}}[f] &= V_{\rm w}'(0^+) =  \max_{i=1,\cdots,\,n} \max_{z_i}\frac{f(z_i)-f(Y_i)}{\|z_i-Y_i\|_p}.
\label{eq:Wass-sensitivity}
\end{align}
\end{proposition}

It follows that when $\varepsilon$ is small
\begin{align*}
V_{\rm w}(\varepsilon) & = \sum_{i=1}^np_i f(Y_i) + \varepsilon \Big(\max_{i=1,\cdots,\,n} \max_{z_i}\frac{f(z_i)-f(Y_i)}{\|z_i-Y_i\|_p}\Big) + o(\varepsilon).
\end{align*}
and that the DRO problem  is (almost) the same as the following mean-sensitivity problem:
\begin{align*}
\lefteqn{\min_x \max_{\gamma \in {\mathcal X}} \int_z f(x,\,z) \Big(\sum_{i=1}^{n} \gamma_i(dz) \Big)} \nonumber \\
& \quad = \min_x  \, {\mathbb E}_{\mathbb{P}}[f(x, Y)] + \varepsilon \Big\{\underbrace{\max_{i=1,\cdots,\,n} \max_{z_i}\frac{f(x, z_i)-f(x, Y_i)}{\|z_i-Y_i\|_p}}_{{\mathcal S}_{\mathbb{P}} [f(x,\cdot)]}\Big\} + o(\varepsilon).
\end{align*}
\medskip

\begin{example}
If $f(z)$ is concave in $z$, then the set of optimal dual variables of \eqref{eq:D1} is
\begin{align*}
\Big\{\lambda\,\big|\, \lambda \geq \max_{i=1,\cdots,\,n}\|\nabla f(Y_i)\|_q\Big\}
= \argmin_{\lambda\geq 0} \sum_{i=1}^k p_i\max_{z_i} \Big\{f(x, z_i)-f(x, Y_i)- \lambda\|z_i-Y_i\|_p\Big\}.
\end{align*}
and worst-case sensitivity
\begin{align}
{\mathcal S}_{\mathbb{P}}[f] = \max_{i=1,\cdots,\,n}\|\nabla f(Y_i)\|_q.
\label{eq:sensitivity concave}
\end{align}
\end{example}

\begin{example}
\label{ex:Wass-inv}
Consider the cost function
\begin{align}
f(x, Y)= -r \min\{x,  Y\} - q\max(x-Y, 0)+ s\max(Y-x, 0) + c x
\label{eq:inv2}
\end{align}
where $0\leq q<c<r$ and $s \geq 0$. The negative of this cost function is the reward function for an inventory problem, so minimizing ${\mathbb E}_{\mathbb P}[f(x, Y)]$ is equivalent to maximizing expected reward. If $x \in (\min_i Y_i, \max_i Y_i)$, it can be shown that for a Wasserstein metric with $p=1$
\begin{align*}
{\mathcal S}_{\mathbb P} [f(x, \cdot)] = \max\{r-q, s\},
\end{align*}
so  the SAA optimizer is also the solution of the DRO problem for a large range of $\varepsilon$, beyond which, the order quantity is either smaller than $\min_i Y_i$ or larger than $\max_i Y_i$, which is not sensible. This suggests that the Wasserstein uncertainty set with $p=1$ may not be a good choice for the robust inventory problem.
\end{example}

\section{Examples}
\label{sec:Examples}

\subsection{Inventory control}
\label{sec:inv}
Consider once again the inventory cost function \eqref{eq:inv2}.
In this experiment, we generated $n=100$ demand realizations $\{Y_1,\cdots,\,Y_n\}$ by sampling from a mixture of two exponential distributions with means $\mu_L=10$ and $\mu_H=100$, where the probability of a sample from population $L$ is 0.9. We assume $r=10$, $c=2$, $q=0$ and $s=4$.
Note that $\varepsilon=0$ is equivalent to SAA.

We begin by comparing the solutions of SAA, the robust inventory problem with a ``budgeted" uncertainty set \cite{gotoh2007newsvendor} ($\varepsilon=0.45$), and the robust problem with a modified $\chi^2$ uncertainty set ($\varepsilon=1.7$). Note first that the order quantity under ``budgeted" uncertainty ($x(\varepsilon)=18$) is smaller than SAA ($x(0)=24$), while that for the modified $\chi^2$  uncertainty set ($x(\varepsilon)=44)$ is larger. Indeed, the worst-case expected cost of the ``budgeted"  solution ($x(\varepsilon)=18$) is larger than that of SAA when evaluated under the modified $\chi^2$ uncertainty set (and vice versa); in the eyes of the modified $\chi^2$ DRO model, solutions of the ``budgeted" DRO problem are less robust than SAA.

These observations can be explained by considering the  worst-case sensitivity associated with each uncertainty set. For  ``budgeted" uncertainty, worst-case sensitivity is the spread of the ``good" part of the reward distribution \eqref{eq:sen_cvar}, while for modified $\chi^2$ (and any smooth $\phi$-divergence) it is the standard deviation \eqref{eq:phi-sensitivity} which depends on the entire cost distribution. For ``budgeted" uncertainty, DRO trades off expected cost in return for a reduction in the spread of the good side of the reward distribution, which is achieved by reducing the order quantity. This comes at the cost of a longer right tail, but this does not affect the sensitivity measure; see Figure \ref{fig:CVaR_reward_distributions2}. On the other hand, the modified $\chi^2$ robust optimizer controls the standard deviation of the cost distribution; a larger order quantity reduces the standard deviation of the distribution and the length of the right tail, but the width of the body increases.

\begin{figure}[h]
\centering
\includegraphics[scale=0.3]{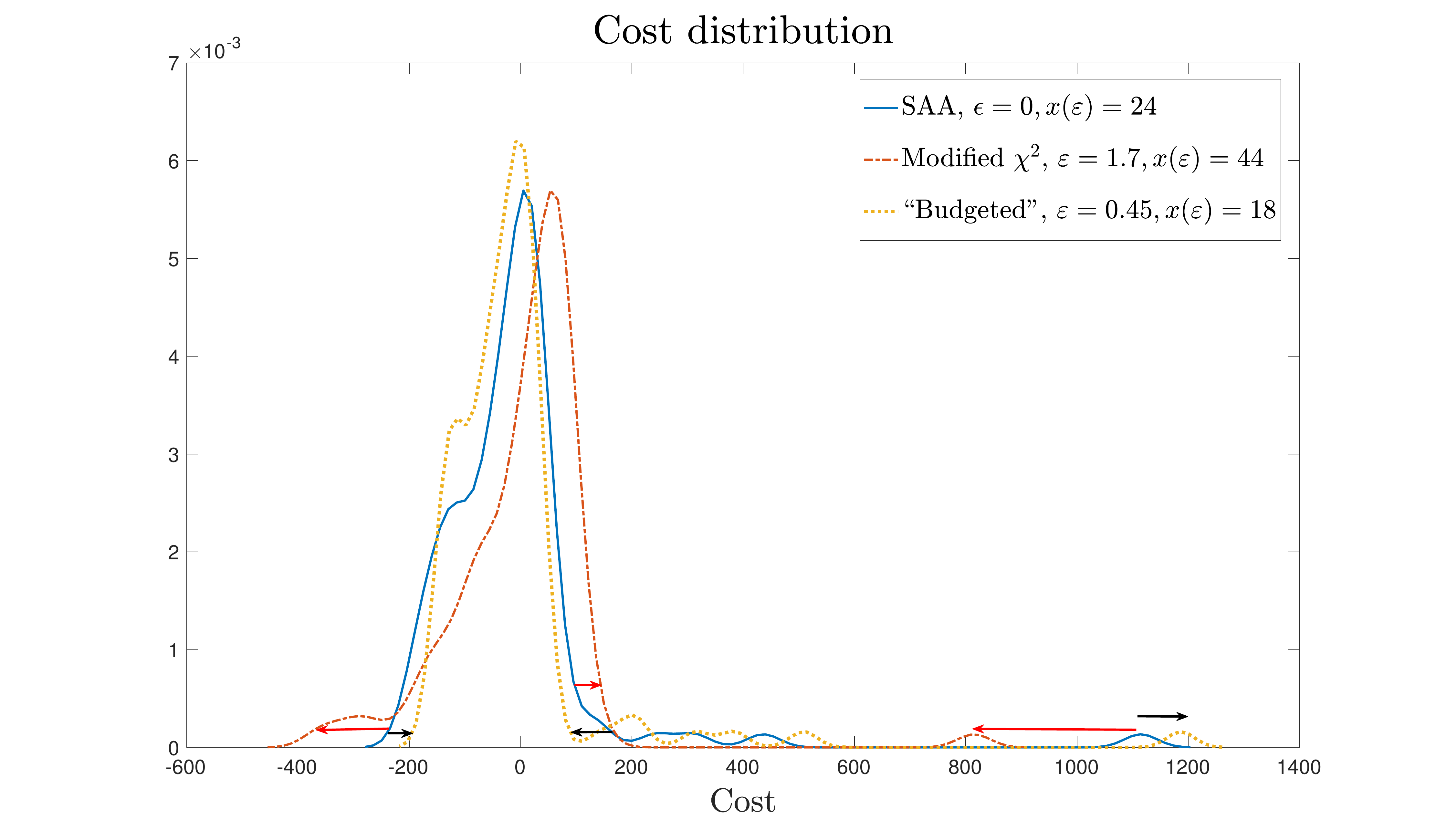}
\caption{This plot shows cost distributions under the SAA solution and the solutions of DRO problems with ``budgeted" and modified $\chi^2$ uncertainty sets. Worst-case sensitivity for a ``budgeted" uncertainty set is the spread of the ``good" side of the cost distribution. The DRO solution decreases the order quantity (relative to SAA) as this reduces the width of the body of the distribution (and hence the sensitivity), but this comes at the cost of a longer right tail; this is indicated by the black arrows. On the other hand, worst-case sensitivity for the modified $\chi^2$ uncertainty set is the standard deviation. The associated DRO order quantity is larger than SAA, which reduces the length of the right tail at the cost of a fatter body (red arrows). From the perspective of the modified $\chi^2$ DRO model, the ``budgeted" solution is less robust than SAA (and vice versa); the worst-case expected cost is larger than the SAA solution and it has a larger worst-case sensitivity.}
\label{fig:CVaR_reward_distributions2}
\end{figure}

We now compare solutions generated by all the uncertainty sets we have discussed. Let $\{x^u(\epsilon)\,\vert\,\epsilon\geq 0\}$ denote the family of worst-case solutions where the uncertainty set $u$ is either KL-divergence, modified $\chi^2$-deviation, Total variation, ``budgeted" uncertainty, or the convex-combination uncertainty set. In Figure \ref{fig:frontiers-inv1} (a), we plot mean-sensitivity frontiers for each family of solution where sensitivity is measured by the standard deviation \eqref{eq:phi-sensitivity} (worst-case sensitivity for smooth $\phi$-divergence). The other plots show frontiers with sensitivity associated with (b) Total Variation \eqref{eq:sen_tv}, (c) ``budgeted uncertainty" \eqref{eq:sen_cvar}, and (d) $CVaR$-deviation of the cost \eqref{eq:sen-CVaRdev}.

While DRO solutions reduce the worst-case sensitivity corresponding to its uncertainty set, they may increase other measures of worst-case sensitivity. When the shortage cost $s=4$, ``budgeted" uncertainty produces solutions that increase the other measures of worst-case sensitivity (and vice versa). This can be seen in Figure \ref{fig:frontiers-inv1} where frontiers for ``budgeted" uncertainty moves in the opposite direction to the others. Concurrently, robust order quantities are decreasing in $\varepsilon$ while those for other uncertainty sets are increasing.  When  $s=0$, however,  the order quantities  from all DRO models are decreasing in $\epsilon$ and the resulting frontiers are the same, as seen in Figure \ref{fig:frontiers-inv2}. Finally, as shown in Example \ref{ex:Wass-inv}, the Wasserstein sensitivity is independent of the order quantity when $p=1$, so the DRO solution is the SAA optimizer.

\begin{figure}[h]
\includegraphics[scale=0.5]{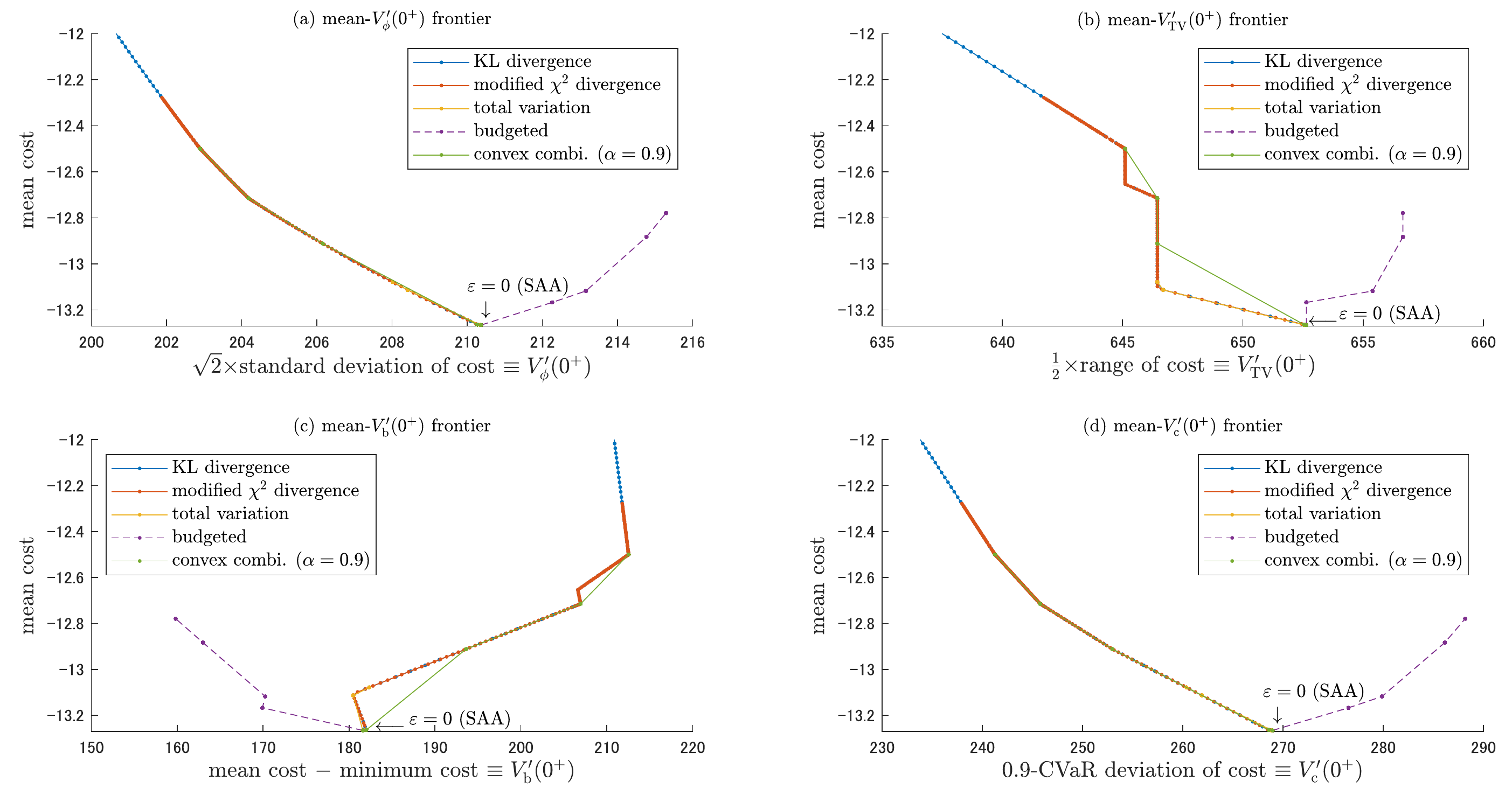}
\caption{Mean-sensitivity frontiers  for the newsvendor problem when $s=4$}
\label{fig:frontiers-inv1}
\end{figure}

\begin{figure}
\includegraphics[scale=0.5]{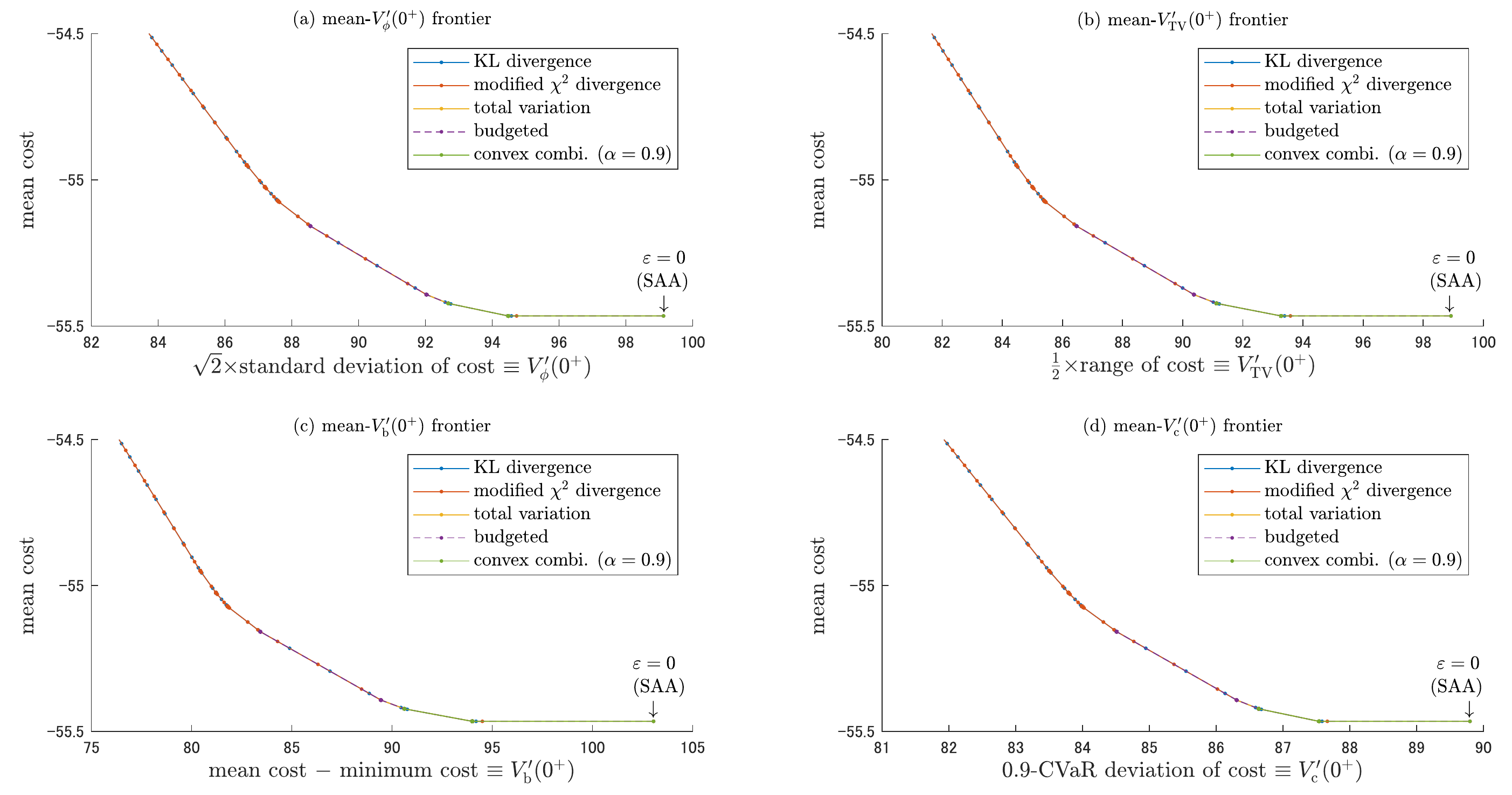}
\caption{Mean-sensitivity frontiers for the newsvendor problem when $s=0$}
\label{fig:frontiers-inv2}
\end{figure}

Mean-sensitivity frontiers can also be used to select the size $\varepsilon$ of an uncertainty set. For example, we can use (a) from Figure \ref{fig:frontiers-inv1} and the modified $\chi^2$ frontier to select $\varepsilon$ for a DRO model with a modified $\chi^2$ uncertainty set (the other frontiers are not needed).  More importantly,  the uncertainty set is an important modelling choice in a DRO model as it determines the measure of sensitivity that is being controlled when solving the worst-case problem.

\subsection{Logistic regression}
\label{sec:LR}
We now consider a higher dimensional example.
 Fig.\,\ref{fig:lr:hearfail:frontiers} shows the four mean-sensitivity frontiers of the seven DROs for the logistic regression using the heart failure clinical records dataset \cite{chicco2020machine}, which consists of 299 samples having 12 covariates.

The ordinary logistic regression is SAA where the cost of the $i$-th sample is defined as
\begin{align*}
f(\mathsf{w},(y_i,\mathsf{x}_i))&:=\ln\big(1+\exp(-y_i\mathsf{w}^\top\mathsf{x}_i)\big),
\end{align*}
where $\mathsf{x}_i\in\mathbb{R}^{d}$ and $y_i\in\{\pm 1\}$ denote the input vector and binary label of the $i$-th sample, respectively, and $\mathsf{w}\in\mathbb{R}^{d}$ is the vector of decision variables. (For the heart failure dataset, the all-one vector is added to the covariates, and thus $d=13$.) For details on the Wasserstein DRO model for this application, see Appendix \ref{App:LR}. For this example, all frontiers in Figure \ref{fig:lr:hearfail:frontiers} are similar, which suggests that solutions produced by the seven DRO models we consider are similar.

\begin{figure}[h]
\includegraphics[scale=0.65]{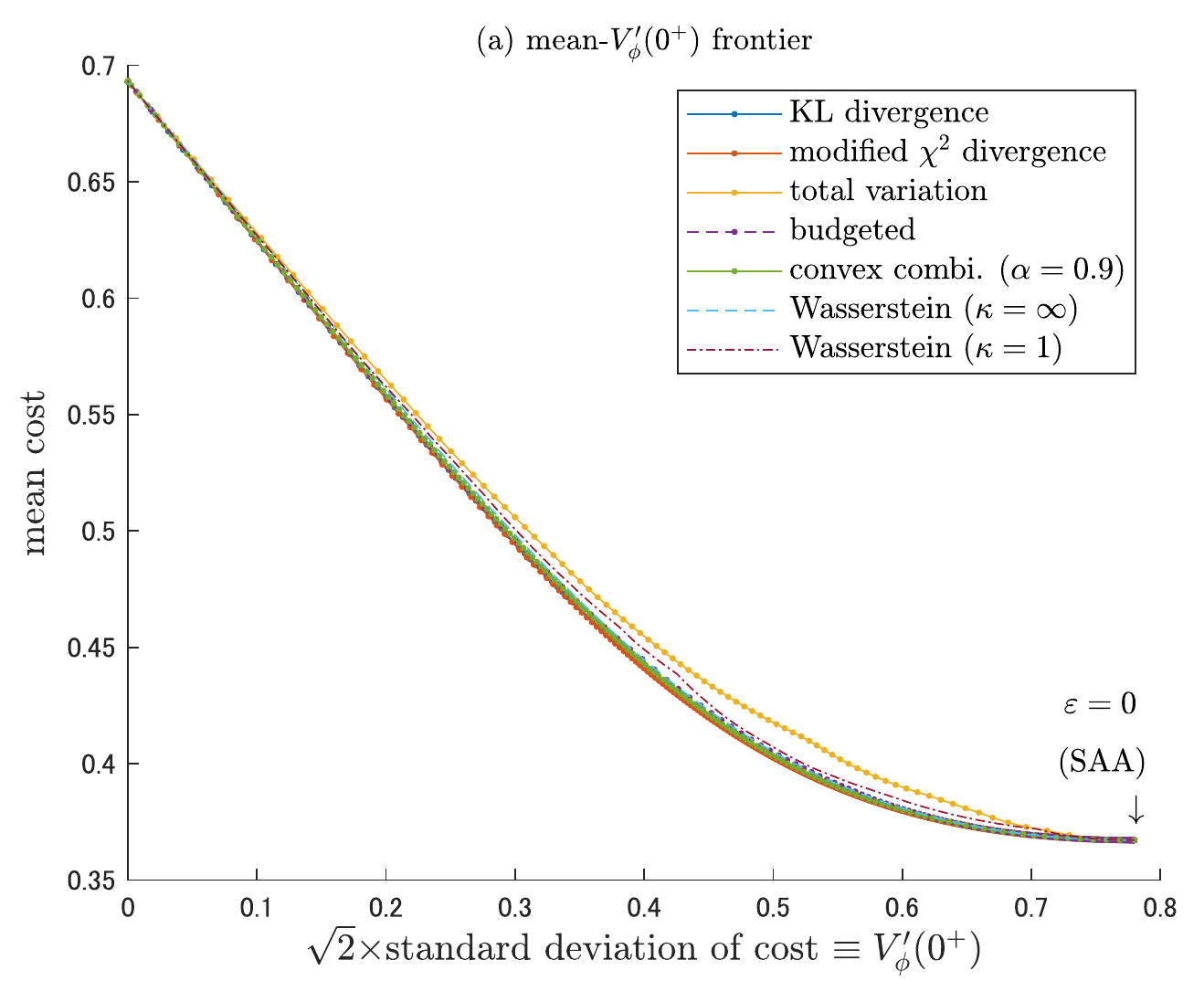}
\caption{Mean-sensitivity frontiers of solutions of the robust logistic regression problem with the Heart Failure data set. Each frontier corresponds to the family of DRO solutions for one of the seven uncertainty sets. Worst-case sensitivity in all the frontier plots is measured by $\sqrt{2}\times{\mbox{standard deviation}}$ (worst-case sensitivity of the modified $\chi^2$ uncertainty set) of the associated in-sample cost distribution. }
\label{fig:lr:hearfail:frontiers}
\end{figure}


\section{Conclusion}

We have derived worst-case sensitivities for the expected reward for several commonly used worst-case models. The results are summarized in Table \ref{table:summary}.

Worst-case sensitivity is a measure of the spread of the cost distribution $f(x, Y)$.
Mathematically, this is a consequence of classical duality results for DRO.
Intuitively, the expected cost under the nominal distribution is sensitive to changes in the probability assigned to extreme values so a cost distribution with a large spread is sensitive to misspecification and not robust. DRO is a tradeoff between  expected cost and sensitivity, where the measure of sensitivity, and hence the resulting solution, depends on the uncertainty set.

Practically, our analysis provides a list of sensitivity measures, each corresponding to a different uncertainty set, that can be used to compare the worst-case cost sensitivity for different decisions. That is, worst-case sensitivity serves as a quantitative  measure of robustness, which to our knowledge has not been proposed in the literature. The expressions for worst-case sensitivity make explicit the measure of spread that DRO is trying to control for each uncertainty set, and can be used to select the uncertainty set for a given application. Mean-sensitivity frontiers can also be used to select its size.

\paragraph{\bf Acknowledgments}
J. Gotoh is supported in part by JSPS KAKENHI Grant 19H02379, 19H00808, and 20H00285. He also thanks Mr. Koichi Fujii for his instruction for the use of RNUOPT, a numerical optimization package of NTT DATA Mathematical Systems Inc., with which the DRO examples in Section \ref{sec:Examples} are solved.
\newpage

\bibliographystyle{plainnat}
\bibliography{references}

\newpage

\appendix

\section{Proofs from Section \ref{sec:DRO_sensitivity}}

\subsection{Proof of Proposition \ref{prop:deviation}}

Without loss of generality, we assume ${\mathbb E}_{\mathbb P}[f(x,Y)]=0$ and  that the $f_i$'s are ordered such that $f_1\geq f_2 \geq \cdots \geq f_n$.

\subsubsection*{${\mathcal A}(\varepsilon;f)$ and ${\mathcal S}(\varepsilon;f)$ are generalized deviation measures}
Clearly, ${\mathcal A}(\varepsilon;f)$ satisfies properties (2) and (3) of Definition \ref{def:deviation}. As for property (1), ${\mathcal A}(\varepsilon;f)\geq 0$ follows from the observation that the nominal measure ${\mathsf p}\in{\mathcal Q}(\varepsilon)$ is feasible so
\begin{align*}
{\mathcal A}(\varepsilon;f) \geq \sum_{i=1}^n p_i f_i= 0.
\end{align*}
It is also clear that ${\mathcal A}(\varepsilon;f)=0$ if $f\equiv 0$. We now show the converse.

Suppose that ${\mathcal A}(\varepsilon;f)=0$ but that $f$ is not equal to $0$. Since $f$ is non-zero with ${\mathbb E}_{\mathbb P}[f]=0$, $f_1 > 0 > f_n$. Let $\mathsf{e}^{(i)}$ denote an $n$-dimensional vector with $1$ in the $i^{th}$ entry and $0$'s elsewhere. Clearly, we can select $\delta>0$ such that
\begin{align*}
{\mathsf q} = {\mathsf p} + \delta[\mathsf{e}^{(1)}-\mathsf{e}^{(n)}] > 0
\end{align*}
and $d\big({\mathsf p} + \delta[\mathsf{e}^{(1)}-\mathsf{e}^{(n)}] \big| {\mathsf p}\big)<\varepsilon$, so
\begin{align*}
{\mathcal A}(\varepsilon;f) \geq \sum_{i=1}^n q_i f_i = \delta (f_1-f_n) > 0
\end{align*}
which contradicts our initial assumption that ${\mathcal A}(\varepsilon;f)=0$. It follows that $f$ must equal $0$.

Since ${\mathcal A}(\varepsilon;f)$ is a measure of deviation, so too is ${\mathcal S}(\varepsilon;f)$.

\subsubsection*{${\mathcal S}_{\mathsf p}(\mathsf{f})$ is a generalized measure of deviation}
We first establish the growth condition on ${\mathcal A}(\varepsilon;f)$. Suppose there is a constant $k>0$ such that the continuity condition \eqref{eq:continuity d} holds for every $\Delta\in{\mathbb R}^n$ satisfying ${\mathsf 1}'\Delta=0$. Given $f$ and $\varepsilon>0$, let ${\mathsf q}(\varepsilon)$ denote the solution of the worst-case problem \eqref{eq:V}. Since $d({\mathsf q}(\varepsilon) | {\mathsf p}) = \varepsilon$, it follows that ${\mathsf q}(\varepsilon)-{\mathsf p} =O(\varepsilon^{\frac{1}{k}})$ so
\begin{align*}
{\mathcal A}(\varepsilon;f) = \sum_{i=1}^n q_i(\varepsilon)f_i \sim O(\varepsilon^{\frac{1}{k}}).
\end{align*}
We now show that ${\mathcal S}_{\mathsf p}(\mathsf{f})$ is a generalized measure of risk. Clearly, ${\mathcal S}_{\mathsf p}(\mathsf{f})$ satisfies conditions (2) and (3) of Definition \ref{def:deviation}.
As for Condition (1), it is easy to show that ${\mathcal S}_{\mathsf p}(\mathsf{f})\geq 0$ and that ${\mathcal S}_{\mathsf p}(\mathsf{f})=0$ if $f$ is equal to $0$; these properties follow from those of ${\mathcal A}(\varepsilon;f)$. To prove the converse,  suppose that
\begin{align}
{\mathcal S}_{\mathsf p}(\mathsf{f}) := \lim_{\varepsilon \downarrow 0}\frac{{\mathcal A}(\varepsilon;f)}{\varepsilon^{\frac{1}{k}}} = 0
\label{eq:contradiction1}
\end{align}
but that $f
$ is not equal to $0$.
Since $d({\mathsf q}|{\mathsf p})$ is continuous in ${\mathsf q}$,  there is a constant $\delta>0$ such that
\begin{align*}
{\mathsf q}(\varepsilon) \equiv \big(q_1(\varepsilon),\cdots, q_n(\varepsilon)\big)^\top := {\mathsf p} + \varepsilon^\frac{1}{k}\delta [\mathsf{e}^{(1)}-\mathsf{e}^{(n)}] > 0
\end{align*}
and
\begin{align*}
d\Big({\mathsf p} + \varepsilon^\frac{1}{k}\delta [\mathsf{e}^{(1)}-\mathsf{e}^{(n)}] \Big| {\mathsf p}\Big) \leq \varepsilon
\end{align*}
for all $\varepsilon>0$ sufficiently small.
In particular, $q_1(\varepsilon) = p_1+ \varepsilon^\frac{1}{k}\delta > 0$, $q_n(\varepsilon) = p_n -  \varepsilon^\frac{1}{k}\delta > 0$, and $q_i(\varepsilon) = p_i(\varepsilon)$ for $i=2,\cdots,\,n-1$.
Since ${\mathsf q}(\varepsilon)\in{\mathcal Q}(\varepsilon)$, it follows that
\begin{align*}
0  <  \sum_{i=1}^n q_i(\varepsilon) f_i
 =   \varepsilon^\frac{1}{k} \delta\Big(f_1-f_n\Big)  \leq  {\mathcal A}(\varepsilon)
\end{align*}
so
\begin{align*}
\frac{{\mathcal A}(\varepsilon)}{ \varepsilon^\frac{1}{k}} \geq \delta\Big(f_1-f_n\Big)>0.
\end{align*}
This contradicts \eqref{eq:contradiction1}.

\section{Proofs from Section \ref{sec:WCS}}

\subsection{$\phi$-divergence uncertainty sets: Proof of Proposition \ref{prop:phi-div}}
\label{App:phi-div}
\begin{align*}
V(\varepsilon)  = & \min_{\delta\geq 0,\, c} \; \frac{1}{\delta} \sum_{i=1}^{n} p_i \max_{q_i} \Big\{ \frac{q_i}{p_i}\delta \big(f_i+c\big) - \phi\Big(\frac{q_i}{p_i}\Big)\Big\} + \frac{\varepsilon}{\delta}-c \\
= &  \sum_{i=1}^n p_i f_i + \min_{\delta\geq 0,\, c} \frac{1}{\delta} \Big\{\sum_{i=1}^{n} p_i\Big[ \phi^*\Big(\delta \big(f_i+c\big)\Big) - \delta(f_i+c)\Big]+ {\varepsilon}\Big\}
\end{align*}
where
\begin{align*}
\phi^*(\zeta) = \max_z \big\{\zeta z - \phi(z)\big\}
\end{align*}
is the convex conjugate of $\phi(z)$ and
\begin{align}
q_i = p_i[\phi']^{-1}\Big(\delta(f_i+c)\Big) = \argmax_{q_i} \Big\{ \frac{q_i}{p_i}\delta \big(f_i+c\big) - \phi\Big(\frac{q_i}{p_i}\Big)\Big\}
\label{eq:wcq_phi_temp}
\end{align}
is the optimizer in the first equality. The worst-case measure ${\mathsf q}(\varepsilon) = \big(q_1(\varepsilon),\cdots,\,q_n(\varepsilon)\big)^\top$ is obtained by substituting the optimizers over $\delta$ and $c$.

Under the assumptions about $\phi(z)$ from \cite{gotoh2018robust}, $\phi^*(\zeta)$ is convex and twice continuously differentiable with
\begin{align}
\phi^*(\zeta)=\zeta + \frac{\zeta^2}{2\phi''(1)}+o(\zeta^2)
\label{eq:phi*-expansion}
\end{align}

Differentiating with respect to $c$ and $\delta$, the first order conditions are
\begin{align}
\sum_{i=1}^{n}p_i[\phi^*]'\Big(\delta[f_i+c]\Big)   &=   1 \label{eq:foc-c} \\
\sum_{i=1}^{n}p_i\Big\{\phi^*\Big(\delta[f_i+c]\Big) -  [\phi^*]'\Big(\delta[f_i+c]\Big)\delta[f_i+c]\Big\}+\varepsilon  &=  0. \label{eq:foc-delta}
\end{align}
Clearly $c(0) = -{\mathbb E}_{{\mathsf p}}(\mathsf{f})$. We can apply the Implicit Function Theorem to show that $c(\delta)$ is continuously differentiable in $\delta$ in the neighborhood of $\delta=0$ so it follows from \eqref{eq:phi*-expansion} and \eqref{eq:foc-c} that
\begin{align*}
c(\delta) = -{\mathbb E}_{\mathsf{p}}(\mathsf{f}) + O(\delta).
\end{align*}
Together with the expansion of $\phi^*(\zeta)$, \eqref{eq:foc-delta} becomes
\begin{align*}
\frac{\delta^2}{2 \phi^{''}(1)}\sum_{i=1}^{n} p_i\big(f_i-{\mathbb E}_{\mathsf{p}}(\mathsf{f})\big)^2 + o(\delta^2) = \varepsilon.
\end{align*}
The Implicit Function Theorem can again be used to show that $\delta(\varepsilon)$ is continuously differentiable on some open interval\footnote{Consider \begin{align*}\frac{y}{2 \phi^{''}(1)}\sum_{i=1}^{n} p_i\big(f_i-{\mathbb E}_{\mathsf{p}}(\mathsf{f})\big)^2 + o(y) = \varepsilon
\end{align*} Observe that $y=0$ when $\varepsilon=0$. The Implicit Function Theorem implies that $y(\varepsilon)$ is continuously differentiable on an open interval containing $\varepsilon=0$. Continuous differentiability of $\delta(\varepsilon)$ on some open interval $(0,\,b)$ follows from the observation that $\delta(\varepsilon)=\sqrt{y(\varepsilon)}$ for $\varepsilon>0$.} $(0,\,b)$ and given by \eqref{eq:delta_phi}.
It follows that  $c(\varepsilon)$ is given by \eqref{eq:c_phi}.

Finally, the worst-case distribution is obtained by substituting \eqref{eq:c_phi} and \eqref{eq:delta_phi} into \eqref{eq:wcq_phi_temp}.

\subsection{Proof of Lemma \ref{lemma:sen_tv}}
\label{App:TV}

\begin{proof}
The worst-case objective is written by the following optimization problem:
\[
\begin{array}{ll}
\underset{\mathsf{q}}{\mbox{maximize}}&\mathsf{f}^\top\mathsf{q}\\
\mbox{subject to}&\mathsf{1}^\top\mathsf{q}=1,\\
                 &\mathsf{q}\geq\mathsf{0},\\
                 &\mathsf{1}^\top|\mathsf{p}-\mathsf{q}|\leq\varepsilon,
\end{array}
\]
Note that if $\varepsilon>0$ is sufficiently small, the nonnegativity condition can be omitted (as long as $p_i>0$ for all $i$).
In addition, let us introduce nonnegative vectors $\mathsf{u},\mathsf{v}\geq{0}$ such that $\mathsf{u}-\mathsf{v}=\mathsf{q}-\mathsf{p}$.
(Note that $\mathsf{1}^\top\mathsf{p}=1$.)
So, under the condition of small $\varepsilon$, the worst-case objective can be rewritten by
\[
\begin{array}{ll}
\underset{\mathsf{u},\mathsf{v}}{\mbox{maximize}}&\mathsf{f}^\top\mathsf{u}-\mathsf{f}^\top\mathsf{v}+\mathsf{p}^\top\mathsf{f}\\
\mbox{subject to}&\mathsf{1}^\top\mathsf{u}-\mathsf{1}^\top\mathsf{v}=0,\\
                 &\mathsf{1}^\top\mathsf{u}+\mathsf{1}^\top\mathsf{v}\leq\varepsilon,\\
                 &\mathsf{u},\mathsf{v}\geq\mathsf{0}.
\end{array}
\]
The dual LP problem is derived as
\[
\begin{array}{ll}
\underset{\lambda\geq 0,\theta}{\mbox{minimize}}&\varepsilon\lambda+\mathsf{p}^\top\mathsf{f}\\
\mbox{subject to}&\mathsf{1}\theta+\mathsf{1}\lambda\geq\mathsf{f},\\
                 &-\mathsf{1}\theta+\mathsf{1}\lambda\geq-\mathsf{f},\\
                 &\lambda\geq 0,
\end{array}
\]
which can be reduced to
\[
\begin{array}{ll}
\underset{\theta}{\mbox{minimize}}&\varepsilon\max\{|f_1-\theta|,...,|f_n-\theta|\}+\mathsf{p}^\top\mathsf{f}.
\end{array}
\]
It is easy to see that its optimality attained at $\theta=\frac{f_{(1)}+f_{(n)}}{2}$ with the optimal value being $\frac{\varepsilon(f_{(1)}-f_{(n)})}{2}+\mathsf{p}^\top\mathsf{f}$.

By the strong duality theorem, the worst-case objective turns out to be
\[
\mathsf{p}^\top\mathsf{f}+\frac{\varepsilon(f_{(1)}-f_{(n)})}{2}
=\mathbb{E}_\mathsf{p}(\mathsf{f})+\frac{\varepsilon(\max(\mathsf{f})-\min(\mathsf{f}))}{2},
\]
which is attained by the solution \eqref{eq:worst-case_solution_to_tv}.
Consequently, we complete the proof of Lemma \ref{lemma:sen_tv}.
\end{proof}

\subsection{Proof of Proposition \ref{prop:sen_cvar_env}}
\label{App:CVaR}
\begin{proof}
Solving 
LP \eqref{eq:dual_cvar} (by a greedy algorithm), a worst-case distribution is given as
\begin{equation}
q_{(i)}=
\left\{
\begin{array}{ll}
\displaystyle(1+\varepsilon)p_{(i)},&i=1,...,k,\\
\displaystyle 1-(1+\varepsilon)\sum_{i=1}^{k}p_{(i)},&i=k+1,\\
\displaystyle 0,&i=k+2,...,n,
\end{array}
\right.
\label{eq:worst-case_dist_cvar}
\end{equation}
where $q_{(i)}$ denotes the probability mass for the $i$-th largest cost $f_{(i)}$.
Putting \eqref{eq:worst-case_dist_cvar} into the objective of \eqref{eq:dual_cvar}$(\equiv \sum_{i=1}^{n}q_{(i)}f_{(i)})$, we have
\begin{eqnarray*}
V_{\rm CVaR}(\varepsilon) &= &\mathrm{CVaR}_{\mathsf{p},\frac{\varepsilon}{1+\varepsilon}}(\mathsf{f}) \\ &= & (1+\varepsilon)\sum_{i=1}^{k}p_{(i)}f_{(i)}+\Big(1-(1+\varepsilon)\sum_{i=1}^{k}p_{(i)}\Big)f_{(k+1)}.
\end{eqnarray*}
With this formula,
\begin{align}
\lefteqn{V_{\rm CVaR}(\varepsilon+\Delta)-V_{\rm CVaR}(\varepsilon)} \nonumber \\
& =  \mathrm{CVaR}_{\mathsf{p},\frac{\varepsilon+\Delta}{1+\varepsilon+\Delta}}(\mathsf{f})-\mathrm{CVaR}_{\mathsf{p},\frac{\varepsilon}{1+\varepsilon}}(\mathsf{f})\nonumber\\
&= (1+\varepsilon+\Delta)\sum_{i=1}^{k}p_{(i)}f_{(i)}+\Big\{1-(1+\varepsilon+\Delta)\sum_{i=1}^{k}p_{(i)}\Big\}f_{(k+1)}\nonumber\\
&\qquad -(1+\varepsilon)\sum_{i=1}^{k}p_{(i)}f_{(i)}-\Big\{1-(1+\varepsilon)\sum_{i=1}^{k}p_{(i)}\Big\}f_{(k+1)}\nonumber\\
&=\frac{\Delta}{1+\varepsilon}\Big[
(1+\varepsilon)\sum_{i=1}^{k}p_{(i)}f_{(i)}+\Big\{1-(1+\varepsilon)\sum_{i=1}^{k}p_{(i)}\Big\}f_{(k+1)}-\underbrace{f_{(k+1)}}_{\mathrm{VaR}_{\mathsf{p},\frac{\varepsilon}{1+\varepsilon}}(\mathsf{f})}\Big]\label{eq:cvar-cvar}\\
&=\frac{\Delta}{1+\varepsilon}\Big(
\mathrm{CVaR}_{\mathsf{p},\frac{\varepsilon}{1+\varepsilon}}(\mathsf{f})-\mathrm{VaR}_{\mathsf{p},\frac{\varepsilon}{1+\varepsilon}}(\mathsf{f})
\Big),\nonumber
\end{align}
which proves \eqref{eq:sen_cvar_env}.
\end{proof}



\subsection{Proof of Proposition \ref{propo:M-Stdev_as_UB_of_CVaR}}
\begin{proof}
Note 
that the derivation of the tight inequality
\eqref{eq:M-Stdev_as_UB_of_CVaR}
is equivalent to showing the reciprocal of the minimum of the following optimization problem is equal to 
$1/C_{\alpha,n}$
\begin{equation}
\begin{array}{r|ll}
\frac{1}{C_{\alpha,n}}
=&\underset{\mathsf{f}}{\mbox{minimize}}&\displaystyle\frac{\sqrt{\mathbb{V}_{\mathsf{p}}(\mathsf{f})}}{\mathrm{CVaR}_{\mathsf{p},\alpha}(\mathsf{f})-\mathbb{E}_{\mathsf{p}}(\mathsf{f})}\\
&\mbox{subject to}&\mathsf{f}\neq C\mathsf{1}\mbox{ for any }C,
\end{array}
\label{eq:frac1}
\end{equation}
where $\mathsf{p}=\mathsf{1}/n$.
Noting that both the denominator and numerator are positively homogeneous and that 
CVaR is translation invariant and, thus, $\mathrm{CVaR}_{\mathsf{p},\alpha}(\mathsf{f})-\mathbb{E}_{\mathsf{p}}(\mathsf{f})=\mathrm{CVaR}_{\mathsf{p},\alpha}(\mathsf{f}-\mathbb{E}_{\mathsf{p}}(\mathsf{f})\mathsf{1})$, the fractional program \eqref{eq:frac1} is equivalently rewritten as
\[
\begin{array}{r|ll}
\frac{1}{C_{\alpha,n}^2}
=&\underset{\mathsf{f},\mathsf{z}}{\mbox{minimize}}&\displaystyle\frac{1}{n}\mathsf{z}^\top\mathsf{z}\\
&\mbox{subject to}&\mathrm{CVaR}_{\mathsf{p},\alpha}(\mathsf{z})=1,\\
&&\mathsf{z}=\mathsf{f}-\frac{1}{n}\mathsf{1}\mathsf{1}^\top\mathsf{f},~\mathsf{f}\neq C\mathsf{1}\mbox{ for any }C,
\end{array}
\]
where the objective function is squared for the convenience. 
Noting that the final constraint implies $
\mathsf{1}^\top\mathsf{z}=0$,
we consider the following relaxed optimization problem:
\begin{equation}
\begin{array}{r|ll}
&\underset{\mathsf{z}}{\mbox{minimize}}&\frac{1}{n}(z_1^2+\cdots+z_n^2)\\
&\mbox{subject to}&\displaystyle\frac{1}{\kappa}\{z_1+\cdots+z_{k}+(\kappa-k)z_{k+1}\}=1,\\
&&z_1\geq z_2\geq \cdots\geq z_n,\\
&&z_1+\cdots+z_n=0,
\end{array}
\label{eq:relaxed_problem_in_proof}
\end{equation}
where $\kappa:=n(1-\alpha)$ and $k:=\lfloor\kappa\rfloor$.

Let us tentatively consider to remove the inequality constraints from \eqref{eq:relaxed_problem_in_proof}:
\[
\begin{array}{r|ll}
&\underset{\mathsf{z}}{\mbox{minimize}}&\frac{1}{n}(z_1^2+\cdots+z_n^2)\\
&\mbox{subject to}&z_1+\cdots+z_{k}+(\kappa-k)z_{k+1}=\kappa,\\
&&(1-\kappa+k)z_{k+1}+z_{k+2}+\cdots+z_{n}=-\kappa.\\
\end{array}
\]
The Lagrangian of this is given by
\[
L(\mathsf{z},\lambda,\mu)=z_1^2+\cdots+z_n^2
-2\lambda\{z_1+\cdots+z_{k}+(\kappa-k)z_{k+1}-\kappa\}
-2\mu\{(1-\kappa+k)z_{k+1}+z_{k+2}+\cdots+z_{n}+\kappa\}
\]
and the optimality condition is given by
\[\renewcommand{\arraystretch}{1.5}
\left\{
\begin{array}{ll}
\displaystyle\frac{\partial L}{\partial z_j}(\mathsf{z},\lambda,\mu)=0,&j=1,...,n,\\
\displaystyle\frac{\partial L}{\partial \lambda}(\mathsf{z},\lambda,\mu)=0&\leftrightarrow z_1+\cdots+z_{k}+(\kappa-k)z_{k+1}=\kappa,\\
\displaystyle\frac{\partial L}{\partial \mu}(\mathsf{z},\lambda,\mu)=0&\leftrightarrow (1-\kappa+k)z_{k+1}+z_{k+2}+\cdots+z_{n}=-\kappa.\\
\end{array}
\right.
\]
Here the first condition is
\[
\frac{\partial L}{\partial z_j}(\mathsf{z},\lambda,\mu)=
\left\{
\begin{array}{ll}
2z_j-2\lambda=0,&1\leq j\leq k,\\
2z_{k+1}-2(\kappa-k)\lambda-2(1-\kappa+k)\mu=0,&j=k+1,\\
2z_j-2\mu=0,&k+2\leq j\leq n.
\end{array}
\right.
\]
Substituting $z_j$'s derived from this to the last two conditions, 
we have
\[
\left\{
\begin{array}{rrl}
\displaystyle \{k+(\kappa-k)^2\}\lambda&+(\kappa-k)(1-\kappa+k)\mu&=\kappa,\\
\displaystyle (\kappa-k)(1-\kappa+k)\lambda&+\{n-k-1+(1-\kappa+k)^2\}\mu&=-\kappa,
\end{array}
\right.
\]
and
\[
\left\{
\begin{array}{ll}
\displaystyle \lambda=\frac{\kappa(n-\kappa)}{n\{k+(\kappa-k)^2\}-\kappa^2}\\
\displaystyle \mu=\frac{-\kappa^2}{n\{k+(\kappa-k)^2\}-\kappa^2},\\
\end{array}
\right.
\]
so
\begin{equation}
\left\{
\begin{array}{lll}
\displaystyle z_j&\displaystyle =\frac{\kappa(n-\kappa)}{n\{k+(\kappa-k)^2\}-\kappa^2},&j=1,...,k,\\
\displaystyle z_{k+1}&\displaystyle =\frac{-\kappa^2+n\kappa(\kappa-k)}{n\{k+(\kappa-k)^2\}-\kappa^2},&\\
\displaystyle z_j&\displaystyle =\frac{-\kappa^2}{n\{k+(\kappa-k)^2\}-\kappa^2},&j=k+2,...,n,\\
\end{array}
\right.
\label{eq:ratio_kkt_sol}
\end{equation}
and
the objective value is then
\begin{equation}
\frac{1}{n}\mathsf{z}^\top\mathsf{z}=\frac{\kappa^2}{n\{k+(\kappa-k)^2-\kappa^2\}}.
\label{eq:squared_obj_val}
\end{equation}
The vector
$\mathsf{z}=(z_1,...,z_n)^\top$ given by \eqref{eq:ratio_kkt_sol} satisfies 
$z_1\geq\cdots\geq z_n$, and turns out to be optimal to \eqref{eq:relaxed_problem_in_proof} as well.
Furthermore,
for the vector $\mathsf{z}$ given by \eqref{eq:ratio_kkt_sol} and any constant $C$,
we can reproduce a nonconstant vector $\mathsf{f}$ as $\mathsf{f}=\mathsf{z}+C\mathsf{1}$, which satisfies $\mathsf{z}=\mathsf{f}-\frac{1}{n}\mathsf{1}\mathsf{1}^\top\mathsf{f}$. 
Consequently, 
the square root of 
\eqref{eq:squared_obj_val} is the optimal value of
\eqref{eq:frac1}.
\end{proof}

\subsection{Wasserstein uncertainty sets}

\subsubsection*{Preliminaries: Some useful results from convex duality}

We summarize general properties of the dual problem and the relationship between optimal dual variables and super-gradients of the value function for the optimization problem
\begin{eqnarray}
V(\varepsilon) := \max_{x\in{\Omega}} F(x) \; \mbox{subject to:}\; G(x) \leq \varepsilon
\label{eq:convex-general}
\end{eqnarray}
which we apply to \eqref{eq:W2}.
Here, we assume that $F: X \rightarrow \mathbb R$ and $G:X \rightarrow {\mathbb R}$, where $X$ is a vector space and $\Omega$ is a convex subset of $\mathcal X$.
The associated dual problem is
\begin{eqnarray}
D(\varepsilon) := \min_{\lambda\geq 0} \max_{x\in{\mathcal X}} F(x) + \lambda\big\{\varepsilon-G(x)\big\}\equiv \max_{\lambda \geq 0} \Big\{ H(\lambda)+\lambda \varepsilon \Big\}
\label{eq:dual-general}
\end{eqnarray}
where $\lambda \in {\mathbb R}$ is the Lagrange multiplier and
\begin{eqnarray*}
H(\lambda) := \max_{x\in{\Omega}} \Big\{F(x)-\lambda G(x)\Big\}.
\end{eqnarray*}
Note that $H(\lambda)$ is convex in $\lambda$.

We denote the derivative of $H$ at $\lambda$ by $H'(\lambda)$, and the directional derivaties
\begin{eqnarray*}
H'(\lambda^+) & := & \lim_{\delta\downarrow 0}\frac{H(\lambda+\delta)-H(\lambda)}{\delta}\\
H'(\lambda^-) & := & \lim_{\delta\downarrow 0}\frac{H(\lambda-\delta)-H(\lambda)}{\delta}
\end{eqnarray*}
(We refer to $H'(\lambda^+)$ and $H'(\lambda^-)$ as the right and left derivative of $H$ at $\lambda$, respectively). $H(\lambda)$ is convex in $\lambda$ so $H'(\lambda^+)$ is increasing in $\lambda$ and $H'(\lambda^-)$ is decreasing in $\lambda$.

\begin{lemma} \label{lemma:duality_prop}
Consider the optimization problem \eqref{eq:convex-general}. Assume that $F(x)$ is concave, that $G(x)$ is convex and non-negative, that there is an $x\in\Omega$ such that $G(x) = 0$, and that $V(\varepsilon)$ is finite for every $\varepsilon$.
Then $V(\varepsilon)$ is concave and increasing in $\varepsilon \geq 0$, and differentiable at almost every $\varepsilon > 0$. When $\varepsilon>0$, strong duality holds $(V(\varepsilon)=D(\varepsilon))$  and the maximum of the dual problem is achieved. If $\lambda(\varepsilon)$ is a solution of the dual problem corresponding to $\varepsilon>0$, then $\lambda(\varepsilon)$ is a super-gradient of $V$ at $\varepsilon$. If $V$ is differentiable at $\varepsilon>0$, then the dual problem has a unique solution and  $V'(\varepsilon) = \lambda(\varepsilon)$.
\end{lemma}

\begin{proof}
When $\varepsilon>0$ the Lagrange Duality Theorem \cite{luenberger1997optimization} implies that strong duality holds and that there exists a solution $\lambda(\varepsilon)$ of the dual problem.
\end{proof}


\begin{proposition} \label{prop:LM_limit}
Assume that $F(x)$ is concave, that $G(x)$ is convex and non-negative, that there is an $x\in\Omega$ such that $G(x) = 0$, and that $V(\varepsilon)$ is finite for every $\varepsilon$.
Let $\{\varepsilon_i\}$ be sequence of positive numbers such that $\varepsilon_i\downarrow \varepsilon$. Suppose that $V$ is differentiable at $\varepsilon_i$, strong duality holds at $\varepsilon_i$, and there is a solution $\lambda_i\equiv\lambda(\varepsilon_i)$ of the dual problem at $\varepsilon_i$, for every $i$. Then
$\lambda_i \equiv \lambda(\varepsilon_i)$ is increasing. If $\lambda_i \uparrow \lambda^* < \infty$ when $\varepsilon_i\downarrow\varepsilon$, then $\lambda^*$ is a solution of the dual problem at $\varepsilon$.
\end{proposition}

\begin{proof}
Since strong duality holds at $\varepsilon_i$
\begin{eqnarray*}
V(\varepsilon_i) =
\min_{\lambda\geq 0} \Big\{H(\lambda)+\lambda \varepsilon_i\Big\}  = H(\lambda(\varepsilon_i))+\lambda(\varepsilon_i) \varepsilon_i,
\end{eqnarray*}
where $H(\lambda)$ is a convex function of $\lambda\geq 0$, and hence is differentiable at almost every $\lambda\geq 0$.

Consider first the case that $\varepsilon_i\downarrow \varepsilon$ where $\varepsilon > 0$. Since $\lambda_i\equiv \lambda(\varepsilon_i)$ is a solution of the dual problem
\begin{eqnarray}
H'(\lambda_i^+) + \varepsilon_i  \geq 0 \nonumber \\
H'(\lambda_i^-) - \varepsilon_i  \geq 0.
\label{eq:dd}
\end{eqnarray}
Since $\varepsilon_i\downarrow \varepsilon$ and $\lambda_i \uparrow \lambda^*$ as $i\rightarrow\infty$, and the right derivative  is increasing in $\lambda$, it follows that
\begin{eqnarray*}
H'({\lambda^*}^+) + \varepsilon \geq \lim_{i\rightarrow\infty} \Big\{ H'({\lambda_i^+}) + \varepsilon_i \Big\} \geq 0.
\end{eqnarray*}
On the other hand, the left derivative $H'(\lambda^-)$ is left-continuous in $\lambda$ so
\begin{eqnarray*}
H'({\lambda^*}^-) - \varepsilon = \lim_{i\rightarrow\infty} \big\{ H'_-(\lambda(\varepsilon_i)) - \varepsilon_i \Big\} \geq 0.
\end{eqnarray*}
It follows that $\lambda^*$ is optimal for the dual problem at $\varepsilon$.

When $\varepsilon = 0$, we need only consider the derivative from the right at $\lambda(0)$. In particular, it follows from \eqref{eq:dd} that
\begin{eqnarray*}
H'({{\lambda^*}^+}) \geq \lim_{\varepsilon_i \downarrow 0} \Big\{H'_+(\lambda(\varepsilon_i)) - \varepsilon_i\Big\} \geq 0.
\end{eqnarray*}
\end{proof}

\subsubsection{Proof of Proposition \ref{prop:eps0}}
\begin{proof}
For every $\lambda\geq 0$ we have
\begin{eqnarray*}
\sum_{i=1}^k p_i\max_{z_i} \Big\{f(z_i)-f(Y_i)- \lambda\|z_i-Y_i\|_p\Big\} \geq 0.
\end{eqnarray*}
We characterize the optimal dual variables by finding the values of $\lambda$ such that the lower bound is attained.

If
\begin{eqnarray*}
\lambda \geq \max_{i=1,\cdots,\,n} \max_{z_i}\frac{f(z_i)-f(Y_i)}{\|z_i-Y_i\|_p}
\end{eqnarray*}
it follows that
\begin{eqnarray*}
f(z_i)-f(Y_i) - \lambda \|z_i-Y_i\|_p\geq 0
\end{eqnarray*}
for every $i$, with equality when $z_i=Y_i$. It follows that
\begin{eqnarray*}
\sum_{i=1}^k p_i\max_{z_i} \Big\{f(z_i)-f(Y_i)- \lambda\|z_i-Y_i\|_p\Big\}=0
\end{eqnarray*}
so $\lambda$ is a solution of the dual problem.
If
\begin{eqnarray*}
\lambda < \max_{i=1,\cdots,\,n} \max_{z_i}\frac{f(z_i)-f(Y_i)}{\|z_i-Y_i\|_p}
\end{eqnarray*}
then for every $i$ such that
\begin{eqnarray*}
\lambda <  \max_{z_i}\frac{f(z_i)-f(Y_i)}{\|z_i-Y_i\|_p}
\end{eqnarray*}
we can find $z_i$ such that $f(z_i)-f(Y_i) - \lambda \|z_i-Y_i\|_p > 0$. It follows that
\begin{eqnarray*}
\sum_{i=1}^k p_i\max_{z_i} \Big\{f(z_i)-f(Y_i)- \lambda\|z_i-Y_i\|_p\Big\}>0
\end{eqnarray*}
so $\lambda$ is not a solution of the dual problem.
\end{proof}

\subsubsection{Proof of Proposition \ref{prop:Wass_sensitivity}}
\begin{proof}
Suppose that the right derivative of $V$ at $\varepsilon=0$ is finite (i.e., $|V'(0^+)|<\infty$). Let $\{\varepsilon_i\}$ be any sequence such that $\varepsilon_i>0$, $\varepsilon_i\downarrow 0$ and $V$ is differentiable at $\varepsilon_i$. The Lagrange Duality Theorem \cite{luenberger1997optimization} implies that strong duality holds and that there exists a unique solution $\lambda_i\equiv\lambda(\varepsilon_i)$ of the dual problem for each $\varepsilon_i$. Since $\{\lambda(\varepsilon_i)\}$ is an increasing sequence that is bounded above by $0\leq V'(0^+) <\infty$ (we have selected $\{\varepsilon_i\}$ such that $V$ is differentiable at $\varepsilon_i$, so $\lambda(\varepsilon_i)= V'(\varepsilon_i)$; in addition, $0 \leq V'(\varepsilon_i) \leq V'(0^+)$ by the concavity of $V(\varepsilon)$), it converges to a limit that is also bounded by $V'_+(0)$
\begin{eqnarray*}
\lambda^* = \lim_{i\rightarrow\infty} \lambda(\varepsilon_i) =   \lim_{i\rightarrow\infty} V'(\varepsilon_i)  \leq V'(0^+) < \infty.
\end{eqnarray*}
By Proposition \ref{prop:LM_limit}, the limit $\lambda^*$ is a solution of the dual problem at $\varepsilon=0$. Since \begin{eqnarray*}
\left[\max_{i=1,\cdots,\,n} \max_{z_i}\frac{f(z_i)-f(Y_i)}{\|z_i-Y_i\|_p},\,\infty\right)
\end{eqnarray*}
is the set of Lagrange multipliers that solve the dual problem when $\varepsilon=0$,
\begin{eqnarray*}
 V'(0^+) \geq \lambda^*  \geq \max_{i=1,\cdots,\,n} \max_{z_i}\frac{f(z_i)-f(Y_i)}{\|z_i-Y_i\|_p}
\end{eqnarray*}
Together with \eqref{eq:Wass_temp1}, it now follows that worst-case sensitivity satisfies \eqref{eq:Wass-sensitivity}.
\end{proof}

\section{Comments on Wasserstein DRO in Example \ref{sec:LR}}
\label{App:LR}

 For the Wasserstein DRO, \cite{abadeh2015distributionally} shows that when the transportation cost is given by $c((\mathsf{x}_i,y_i),(\mathsf{x}_j,y_j))=\|\mathsf{x}_i-\mathsf{x}_j\|_2+\kappa|y_i-y_j|$ for a constant $\kappa\geq 0$, the robust counterpart can be reduced to the following convex optimization problem:
\[
\begin{array}{lll}
\underset{\lambda\geq 0,\mathsf{s},\mathsf{w}}{\mbox{minimize}}
 & \varepsilon\lambda + \frac{1}{n}\sum\limits_{i=1}^ns_i\\
\mbox{subject to}& s_i\geq \ln\Big(1+\exp\big(-y_i(\mathsf{x}_i^\top\mathsf{w})\big)\Big),             & i=1,...,n,\\
                 & s_i\geq \ln\Big(1+\exp\big(y_i(\mathsf{x}_i^\top\mathsf{w})\big)\Big)-\kappa\lambda,& i=1,...,n,\\
                 & \lambda \geq \|\mathsf{w}\|_2.
\end{array}
\]
Here $\kappa$ balance the cost of the covariates and that of the binary label.
If the user is not interested in the mislabeling, s/he can set as $\kappa=\infty$, so that the formulation results in the regularized logistic regression:
\begin{align*}
\underset{\lambda\geq 0,\mathsf{s},\mathsf{w}}{\mbox{minimize}}\quad
 & \varepsilon\|\mathsf{w}\|_2 + \frac{1}{n}\sum\limits_{i=1}^n\ln\Big(1+\exp\big(-y_i(\mathsf{x}_i^\top\mathsf{w})\big)\Big).
\end{align*}
Then the worst-case sensitivity is given by $\mathcal{S}_{\mathsf{p}}(\mathsf{f})=\|\mathsf{w}_{\rm SAA}\|_2
$ where 
$\mathsf{w}_{\rm SAA}$ is a solution to the ordinary logistic regression. 

\end{document}